\documentclass[11pt]{llncs}
\usepackage{graphicx}
\usepackage{graphics}
\usepackage{epsfig}
\usepackage{epstopdf}
\usepackage{amsmath}
\usepackage{latexsym}
\usepackage{wrapfig}
\usepackage{amsfonts}
\usepackage{cite}

\makeatletter

\newcommand{\Rmnum}[1]{\expandafter\@slowromancap\romannumeral #1@}
\makeatother

\newcommand{\deleted}[1]{}

\begin{document}

\title{Chromatic Clustering in High Dimensional Space}
\author{Hu Ding   \hspace{0.25in} Jinhui Xu}
\institute{
 Department of Computer Science and Engineering\\
State University of New York at Buffalo\\
 \email{\tt \{huding, jinhui\}@buffalo.edu}\\
}
\maketitle

\thispagestyle{empty}

\begin{abstract}
In this paper, we study a new type of clustering problem, called {\em Chromatic Clustering}, in high dimensional space. Chromatic clustering seeks to partition a set of colored points into 
%a new type of clustering problems, called {\em chromatic clustering}.  Chromatic clustering takes as input a set of colored data items and partitions them into 
groups (or clusters) so that no group contains points with the same color and a certain objective function is optimized. In this paper, we consider two variants of the problem, chromatic $k$-means clustering (denoted as $k$-CMeans) and chromatic $k$-medians clustering (denoted as $k$-CMedians), and investigate their hardness and approximation solutions. For $k$-CMeans, we show that 
the additional coloring constraint destroys several key properties (such as the locality property) used in existing $k$-means techniques (for ordinary points), and significantly complicates the problem. 
There is no FPTAS for the chromatic clustering problem, even if $k=2$.  To overcome the additional difficulty, we develop a standalone result, called {\em Simplex Lemma}, which enables us to 
efficiently approximate the mean point of an unknown point set through a fixed dimensional simplex. A nice feature of the simplex is its independence with the dimensionality of the original space, and thus can be used for problems in very high dimensional space. With the simplex lemma, together with several random sampling techniques, we show that a $(1+\epsilon)$-approximation of $k$-CMeans can be achieved in near linear time through a sphere peeling algorithm. For $k$-CMedians, we show that a similar  sphere peeling algorithm exists for achieving constant approximation solutions. 
%Based on several new geometric observations and an interesting sphere peeling approach,  we show that a near linear time (on $n$ and $d$) $(1+\epsilon)$-approximation is, however, still achievable  for the chromatic %clustering problem.  
 \end{abstract}

\newpage

\pagestyle{plain}
\pagenumbering{arabic}
\setcounter{page}{1}
\vspace{-0.3in}
\section{Introduction}
\label{sec-intro}
\vspace{-0.1in}
%\vspace{-0.1in}

Clustering is one of the most fundamental problems in computer science and finds applications in many different areas \cite{AV07,BBC,BHI,AP98,DEF,GG06,HM04,KR99,OSS,BKK,HDX11}. Most existing clustering techniques assume that the to-be-clustered data items are independent from each other.
%\cite{BC03,BHI,AV07,AP98,FG88,G85,HM04,KR99}. 
Thus each data item can ``freely'' determine its membership within the resulting clusters, without paying attention to the clustering of other data items. 
%Popular clustering techniques for such type of data include $k$-means, $k$-center and $k$-median clusterings \cite{BC03,BHI,AV07,AP98,FG88,G85,HM04,KR99}. 
In recent years, there are also considerable attentions on clustering dependent data and a number of clustering techniques, such as correlation clustering, point-set clustering, ensemble clustering, and correlation connected clustering, have been developed \cite{BBC,BKK,DEF,GG06,HDX11}.

In this paper, we consider a new type of clustering problems, called {\em Chromatic Clustering}, for dependent data. Roughly speaking, a chromatic clustering problem takes as input a set of colored data items and groups them into clusters, according to certain objective functions, so that no pair of items with the same color are grouped together (such a requirement is called {\em chromatic constraint}).  Chromatic clustering  captures the mutual exclusiveness relationship among data items and is a rather useful model for various applications.
Due to the additional chromatic constraint, chromatic clustering is thus expected to simultaneously solve the ``coloring'' and clustering problems, which significantly complicates the problem. 
As  it will be shown later, the chromatic clustering problem is challenging to solve even for the case that each color is shared only by  two data items.

For chromatic clustering, we consider in this paper two variants,  {\em Chromatic $k$-means Clustering ($k$-CMeans)} and {\em Chromatic $k$-median Clustering ($k$-CMedians)},  in $\mathbb{R}^{d}$ space, where the 
dimensionality could be very high and $k$ is a fixed number.  
%four variants, Chromatic $k$-means Clustering ($k$-CMeans), Chromatic $k$-Median Clustering ($k$-CMedian), Chromatic $k$-Center Clustering ($k$-CCenter), and Chromatic $k$-Clustering in chromatic Correlation %Graphs ($k$-CCG). 
In both variants, the input is a set $\mathcal{G}$ of $n$ point-sets $G_{1}, \cdots, G_{n}$ with each containing a maximum of $k$ points in $d$-dimensional space, and the objective is to partition all points of $\mathcal{G}$ into $k$ different clusters so that the chromatic constraint is satisfied and the total squared distance  (i.e., $k$-CMeans) or total distance (i.e., $k$-CMedians)  from each point to the center point (i.e., median or mean point) of its cluster is minimized.

\textbf{Motivation:} The chromatic clustering problem is motivated by several interesting applications.  One of them is for determining the topological structure of chromosomes in cell biology \cite{HDX11}. In such applications, a set of 3D probing points (e.g., using BAC probes) is extracted from each homolog of the interested chromosome (see Figure \ref{fig-probe} in Appendix), and the objective is to determine, for each chromosome homolog, the common spatial distribution pattern of the probes among a population of cells. For  this purpose, the set of probes from each homolog is converted into a high dimensional feature point in the feature space, where each dimension represents the distance between a particular pair of probes. Since each chromosome has two (or more as in cancer cells) homologs, each cell contributes $k$  (i.e., two or more) feature points.  
%the interested chromosome is first labeled with a number of BAC probes.  Each probe forms a point in the 3D microscopic image (see Figure. \ref{fig-probe} in Appendix). To determine whether there exists a common spatial distribution pattern of the probes among a population of cells, the set of probes from each homolog of a chromosome is first converted into a high dimensional point in the feature space, where each dimension represents the distance between a particular pair of probes.  Since each chromosome has two (or more as in cancer cells) homologs, the probes from each cell is thus mapped into a set of $k$ points in the feature space, where $k$ is the number of chromosome copies. 
%
Due to technical limitation, it is impossible to  identify the same homolog from all cells.  
%Due to physical limitations,
%%of current imaging and labeling techniques,
%Thus, 
Thus, the $k$ feature points from each cell form a point-set with the same color (meaning that they are undistinguishable). To solve the  problem, one could chromatically cluster all point-sets into $k$ clusters (after normalizing the cell size), with each 
corresponding to a homolog, and use the mean or median point of each cluster as its common pattern. 
%first normalize the size of all chromosome copies and then use
 %$k$-CMeans to find the mean point for each homolog of the chromosome as their preferred topological structure.  
% The chromatic clustering problems can also be used to establish correspondence (a core problem in image registration) in medical imaging and computer vision for similar objects appeared in multiple images, such as brachytherapy seeds \cite{JZM,SMX}, and to solve the rigid structure realization problem.
 
 %also find applications in medicine for solving a  challenging brachytherapy seed localization problem \cite{JZM,SMX}, and rigid structure realization problem.

% and be clustered in different groups when determining the topologies of corresponding chromosome copies.
%As shown in \cite{HDX11}, one way to solve the problem is to formulate it as a chromatic cone clustering problem. However

\textbf{Related works:} As its generalization, chromatic clustering is naturally related to the traditional clustering problem. Due to the additional chromatic constraint, chromatic clustering could behave quite differently from its counterpart. 
%For instance, 
%%With the additional chromatic constraint, the $k$-CMeans and $k$-CMedian problems behaves much differently from the ordinary $k$-means problem. 
%some key properties used in existing $k$-means algorithms no longer hold for $k$-CMeans. Particularly, the locality property of $k$-means
For example, the $k$-means algorithms in \cite{BHI,KSS} relies on the fact that all input points in a Voronoi cell of the optimal $k$ mean points belong to the same cluster. However, such a key locality property no 
longer holds for the $k$-CMeans problem.  

Chromatic clustering falls in the umbrella of clustering with constraint. For such type of clustering,  several solutions exist for some variants \cite{BD06}.  Unfortunately, due to their heuristic nature, none of them can yield quality guaranteed solutions for the chromatic clustering problem. The first quality guaranteed solution for chromatic clustering was obtained recently by Ding and Xu. In \cite{HDX11}, they 
%As a new problem, chromatic clustering has been studied only recently. In \cite{HDX11}, Ding and Xu 
considered a special chromatic clustering problem, where every point-set has exactly $k$ points in the first quadrant, and the objective is to cluster points  by cones apexed at the origin, and presented the first PTAS  for constant $k$. The $k$-CMeans and $k$-CMedians problems considered in this paper are the general cases of the chromatic clustering problem. Very recently, Arkin {\em et al.}  \cite{ADH} considered  a chromatic 2D $2$-center clustering problem and presented both approximation and exact solutions. 
%with chromatic requirement in two dimensional space.

%several other The tutorial \cite{BD06} lists several recent results for constraint clustering problems, including the hardness of them. But as far as we know, there is no neither algorithm with theoretical guarantee, nor hardness proof for chromatic clustering until now. 

 \vspace{-0.1in}
\subsection{Main Results and Techniques}
\vspace{-0.05in}
 
In this paper, we present three main results, a constant approximation and a $(1+\epsilon)$-approximation for $k$-CMeans and their extensions to $k$-CMedians.
  \vspace{-0.05in}
\begin{itemize}
\item {\bf Constant approximation:} We show that given any $c$-approximation for $k$-means clustering, it could yield a $(2ck^2+2k-1)$-approximation for $k$-CMeans. This not only provides a way for us to generate an initial constant approximation  solution for $k$-CMeans through some $k$-means algorithm, but more importantly reveals the intrinsic connection between the two clustering problems.

\item {\bf $(1+\epsilon)$-approximation:} We show that a near linear time $(1+\epsilon)$-approximation solution for $k$-CMeans can be obtained using an interesting sphere peeling algorithm.  Due to the lack of locality property in $k$-CMeans, our sphere peeling algorithm is quite different from the ones used in \cite{KSS,BHI}, which in general do not guarantee a $(1+\epsilon)$-approximation solution for $k$-CMeans as shown by our first result. Our sphere peeling algorithm is based on another standalone result, called {\em Simplex Lemma}. The simplex lemma enables us to obtain an approximate mean point of a set of unknown points through a grid inside a simplex determined by some partial knowledge of the unknown point set. A unique feature of the simplex lemma is that the complexity of the grid is {\em  independent of the dimensionality}, and thus can be used to solve problems in high dimensional space. With the simplex lemma, our sphere peeling  algorithm iteratively generates the mean points of $k$-CMeans with each iteration building a simplex for the mean point.

%
%
%Here, we need specially point out that we can not use the peeling algorithm from \cite{KSS} or \cite{BHI} for $k$-meanss or $k$-medians, since their algorithms do not consider the chromatic requirement. For example, although the algorithm from \cite{KSS} provides $(1+\epsilon)$-approximation for $k$-meanss, it can not guarantee $(1+\epsilon)$-approximation for $k$-CMeans. In order to get the $(1+\epsilon)$-approximation, we first show an independent result, {\em Simplex Lemma}. Generally speaking, the Simplex Lemma shows that we can approximate the mean point of some unknown points set via building a grid inside a simplex. The simplex is depends on some prior knowledges of the unknown points set. More important, the complexity of the grid is independent on the dimensionality, which would help us solve the problem in high dimension space. In the peeling algorithm, we would approximate the mean point of each cluster one by one. Each time, the algorithm constructs a simplex, and find the approximate mean point from the grid inside the simplex. Finally, the total running time is near-linear on the data size, and linear on the dimensionality. 
 
\item {\bf Extensions to $k$-CMedians:} We further extend the idea for  $k$-CMeans to $k$-CMedians. Particularly, we show that any $c$-approximation for $k$-medians can be used to yield a $((2+\epsilon)ck^2+(2+\epsilon)k+1)$-approximation for $k$-CMedians, where the $\epsilon$ error comes from the difficulty of computing the optimal median point (i.e., Fermat Weber point). With this and a similar sphere peeling technique, we obtain a $(5+\epsilon)$-approximation for $k$-CMedians. Note that although $k\ge 2$ is a constant in this paper, a $(5+\epsilon)$-approximation is still much better  than a $((2+\epsilon)ck^2+(2+\epsilon)k+1)$-approximation.  

\end{itemize}

%%%%%%%%%%%% full paper
Due to space limit, many details of our algorithms, proofs, and figures are put in Appendix.

  \vspace{-0.15in} 
\section{Preliminaries}
\label{sec-pre}

  \vspace{-0.1in}
In this section, we introduce some definitions which will be used throughout the paper.  
%$4$-approximation algorithm, and finally present a randomized PTAS for CMean.
%\vspace{-0.05in}

\begin{definition}[Chromatic Partition]
 \label{cpart}
Let $\mathcal{G}=\{G_{1}, \cdots, G_{n}\}$ be a set of $n$ point-sets with each $G_{i}=\{p^{i}_{1}, \dots, p^{i}_{k_{i}}\}$ consisting of $k_{i}\leq k$ points in $\mathbb{R}^d$ space.
%Let $\mathcal{G}=\{G_{1}, \cdots, G_{n}\}$ be a set of sets with each  $G_{i}=\{p^{i}_{1}, \dots, p^{i}_{k_{i}}\}$ consisting of $k_{i}\leq k$ elements (e.g., points and vertices). 
A chromatic partition of $\mathcal{G}$ is a partition of the $\sum_{1\leq i\leq n}k_{i}$ points into $k$ sets, $U_{1}, \cdots, U_{k}$,  such that each $U_{i}$ contains no more than one point from each $G_{j}$ for $j=1, 2, \cdots, n$.
 \end{definition}

\vspace{-0.1in}
\begin{definition}[Chromatic $k$-means Clustering ($k$-CMeans)]
\label{def-kcm}
Let $\mathcal{G}=\{G_{1}, \cdots,$ $G_{n}\}$ be a set of $n$ point-sets with each $G_{i}=\{p^{i}_{1}, \dots, p^{i}_{k_{i}}\}$ consisting of $k_{i}\leq k$ points in $\mathbb{R}^d$ space.
%, where $k=\max_{i=1}^{n} k_{i}$. 
The chromatic  $k$-means clustering (or $k$-CMeans) of $\mathcal{G}$ is to find $k$ points $\{m_{1}, \cdots, m_{k}\}$ in $\mathbb{R}^d$ space and a chromatic partition $U_{1}, \cdots, U_{k}$ of $\mathcal{G}$ such that $\frac{1}{n}\sum_{j}\sum_{q\in U_{j}}||q-m_{j}||^2$ is minimized. The problem is called full $k$-CMeans if $k_1=k_2=\cdots=k_n=k$.
\end{definition}

%\begin{definition}[Full $k$-CMeans]
%\label{def-ckcm}
%Let $\mathcal{G}$ be an instance of $k$-CMeans. If $k_1=k_2=\cdots=k_n=k$, then it is a full $k$-CMeans problem.
%\end{definition}

%From the above definition, we know that full $k$-CMeans is a special case of $k$-CMeans.

For both $k$-CMedians and $k$-CMeans, a problem often encountered in our approach is ``How to find the best cluster for each point in $G_{i}$ if the $k$ mean or median points $A=\{m_1, \cdots, m_k\}$ are already known?'' An easy way to solve this problem is to first build a complete bipartite graph $(G_{i}\cup A, E_{i})$ with points in $G_{i}$ and $A$ as the two partites and then compute a minimum weight bipartite matching as the solution, where the edge weight is the Euclidean distance or squared distance of the two corresponding vertices. Clearly, this can be done in a total of $O(k^{3}dn)$ time for all $G_{i}$'s.  (We call this procedure as {\bf bipartite matching}.)

\vspace{-0.15in}
\section{Hardness of $k$-CMeans}
\label{sec-hard}
\vspace{-0.1in}

It is easy to see that $k$-means is a special case of $k$-CMeans (i.e., each $G_i$ contains exactly one point). As shown by Dasgupta \cite{D08},  $k$-means in high dimensional space is NP-hard even if $k=2$.  Thus, we immediately have the following theorem.

  \vspace{-0.05in}

\begin{theorem}
\label{the-fptas}
$k$-CMeans  is NP-hard for $k\ge 2$ in high dimensional space.
\end{theorem}
  \vspace{-0.3in}
  
  \subsection{Is Full $k$-CMeans Easier?}
  \vspace{-0.05in}

%From Theorem \ref{the-fptas}, we know the hardness of $k$-CMeans. 
It is interesting to know whether full $k$-CMeans is  easier than general $k$-CMeans, since it is disjoint with $k$-means when $k\geq 2$. %So we can not easily figure out whether Full $k$-CMeans is easier than the general $k$-CMeans problem. 
The following theorem gives a negative answer to this question.

  \vspace{-0.1in}
\begin{theorem}
\label{mfptas}
Full $k$-CMeans  is NP-hard and has no FPTAS for $k\ge 2$ in high dimensional space  unless P=NP (see Appendix for the proof).
\end{theorem}
  \vspace{-0.1in}
 
 %The main idea of Theorem \ref{mfptas} is to modify the example construct by \cite{D08} to be an instance of Full $k$-CMeans. See the proof in Appendix (Section \ref{sec-mfptas}).

The above theorem indicates that the fullness of $k$-CMeans does not reduce the hardness of the problem. However, this does not necessarily mean that full $k$-CMeans is as  difficult as general $k$-CMeans to achieve a $(1+\epsilon)$-approximation for fixed $k$. Below we show that a $(1+\epsilon)$-approximation can be relatively easily achieved for full $k$-CMeans through some random sampling technique. 

%However, the $(1+\epsilon)$-approximation for Full $k$-CMeans is not difficult to get. 

First we introduce a key lemma from \cite{IKI}. %for random sampling. 
Let $S$ be a set of $n$ points in $\mathbb{R}^d$ space, $T$ be a randomly selected subset from $S$ with $t$ points, and $\overline{x}(S)$, $\overline{x}(T)$ be the mean points of $S$ and $T$ respectively.

  \vspace{-0.1in}
\begin{lemma}[\cite{IKI}]
\label{lem-dis}
With probability $1-\eta$, $||\overline{x}(S)-\overline{x}(T)||^2<\frac{1}{\eta t}Var^{0}(S)$, where $Var^{0}(S)$ $=(\sum_{s\in S}||s-\overline{x}(S)||^2)/n$. 
\end{lemma}
  \vspace{-0.1in}
%From the above lemma, we immediately have the following lemma.

  \vspace{-0.15in}
  \begin{lemma}
\label{lem-select}
Let $S$ be a set of elements, and $S'$ be a subset of $S$ such that 
%Given a set $S$ with $n$ element, and $S'\subset S$,
$\frac{|S'|}{|S|}=\alpha$. If randomly select $\frac{t\ln\frac{t}{\eta}}{\ln(1+\alpha)}=O(\frac{t}{\alpha}\ln\frac{t}{\eta})$ elements from $S$, with probability at least $1-\eta$, the sample contains at least $t$ elements from $S'$.
\end{lemma}
\vspace{-0.2in}
\begin{proof}
If we randomly select $z$ elements from $S$, then it is easy to know that with probability $1-(1-\alpha)^z$, there is at least one element from the sample belonging to $S'$. If we want the probability  $1-(1-\alpha)^z$  equal to $1-\eta/t$, $z$ has to be $\frac{\ln\frac{t}{\eta}}{\ln\frac{1}{1-\alpha}}=\frac{\ln\frac{t}{\eta}}{\ln(1+\frac{\alpha}{1-\alpha})}\leq\frac{\ln\frac{t}{\eta}}{\ln(1+\alpha)}=O(\frac{1}{\alpha}\ln\frac{t}{\eta})$ (by Taylor series and $\alpha<1$, $\ln(1+\alpha)=O(\alpha)$).
Thus if we perform $t$ rounds of random sampling with each round selecting $O(\frac{1}{\alpha}\ln\frac{t}{\eta})$ elements,  we get at least $t$ elements from $S'$ with probability at least $(1-\eta/t)^t\geq 1-\eta$.
%We can consider each element in $S$ as an variable; if it belongs to $S'$, the corresponding variable is equal to $1$, %and otherwise it is  $0$. Then the expected value is equal to $\alpha$, the variance is equal to $\alpha(1-\alpha)$. If %we randomly select $m$ elements from $S$ and there are $T$ elements from $S'$, then by Markov inequality we have
%$$Pr(|\frac{T}{m}-\alpha|\geq\frac{\alpha}{2})\leq\frac{4\alpha(1-\alpha)}{m\alpha^2}.$$
%If we let $m=\frac{4(1-\alpha)}{\alpha\lambda}$, then with probability at least $1-\lambda$ we have %$T\geq\frac{\alpha}{2}m=\frac{2(1-\alpha)}{\lambda}$.
\qed
\end{proof}

Lemma \ref{lem-dis} tells us that if we want to find an approximate mean point  within a distance of $\epsilon Var^{0}(S)$ to the mean point, we just need to take a random sample of size $O(1/\epsilon)$.  Lemma \ref{lem-select} suggests that for any set $S$ and its subset $S'\subset S$ of size $\alpha |S|$, we can have a random subset $T$ of $S'$ with size $O(1/\epsilon)$ by randomly sampling  directly from $S$ $O(\frac{1}{\epsilon\alpha}\ln \frac{1}{\epsilon})$ points, even if $S'$ is an unknown  subset of $S$. 
%from an unknown (or hidden) subset $S'\subset S$ with a fraction $\alpha$
%if there is some hidden subset $S'\subset S$ with a fraction $\alpha$, and we want get a sample from $S'$ with size $O(\frac{1}{\epsilon})$, we just need to take a sample from $S$ with size $O(\frac{\frac{1}{\epsilon}\ln \frac{1}{\epsilon}}{\alpha})$. 
Combining the two lemmas, we can immediately compute an approximation solution for full $k$-CMeans in the following way. First, we note that in full $k$-CMeans, each optimal cluster contains exact $n$ points from the total of $kn$ points in $\mathcal{G}$. This means that each cluster has a fraction of  $\frac{1}{k}$ points from $\mathcal{G}$. Then, we can obtain an approximate mean point for each optimal cluster by (1) randomly sampling $O(\frac{k}{\epsilon}\ln\frac{1}{\epsilon})$ points from $\mathcal{G}$,  (2) enumerating all possible subsets of size $O(1/\epsilon)$ to find the set $T$ which is a random sample of the unknown optimal cluster, and (3) computing the mean of $T$ as the approximate mean point of the optimal cluster.   
Finally, we can generate the $k$ chromatic clusters from the $k$ approximate mean points by using the bipartite matching procedure (see Section \ref{sec-pre}).  %Summarizing the above approach, we have the following theorem for full $k$-CMeans.

%For Full $k$-CMeans problem, each optimal cluster contains $n$ points of the total $kn$ points. That means each cluster has a fraction $\frac{1}{k}$, so we can get the mean point of each optimal cluster via a sample with %size $O(\frac{k}{\epsilon}\ln\frac{1}{\epsilon})$. After got the approximation positions of the $k$ mean points, we can use bipartite matching algorithm to find the proper matching to them for each $G_i$, $1\leq i\leq n$. So %we have the following theorem:
  \vspace{-0.1in}
\begin{theorem}
\label{the-cptas}
With constant probability, a $(1+\epsilon)$-approximation of full $k$-CMeans can be obtained in $O(2^{poly(\frac{k}{\epsilon})}$ $nd)$ time.
%  algorithm for full $k$-CMeans generating  with constant probability. 
\end{theorem}
  \vspace{-0.1in}

With the above theorem, we only need to focus on the general $k$-CMeans problem in the remaining sections. Note that in the general case, some clusters may have a very small fraction (rather than $1/k$)  of points, thus we can not use the above method to solve the general $k$-CMeans problem.
%In the following parts of this paper, we will focus on the general $k$-CMeans problem. 

  \vspace{-0.17in}
\section{Constant Approximation from $k$-means}
\label{sec-constant}
  \vspace{-0.12in}

In this section, we show that a constant approximation solution for $k$-CMeans can be produced from an approximation solution of $k$-means. Below is the main theorem of this section.
%We present the result at the beginning.
\vspace{-0.05in}
\begin{theorem}
\label{the-constant}
Let $\mathcal{G}=\{G_1, \cdots, G_n\}$ be an instance of $k$-CMeans, and $\mathcal{C}$ be the $k$ mean points of a constant $c$-approximation solution of $k$-means on the points $\cup^n_{i=1}G_i$. 
%Assume there is an $c$-approximation solution for $k$-meanss clustering on $\cup^n_{i=1}G_i$ without the chromatic requirement, and $\mathcal{C}$ contains the $k$ mean points of the $k$ clusters. 
Then $[\mathcal{C}]^k$ contains at least one $k$-tuple which could induce a $(2ck^2+2k-1)$-approximation of $k$-CMeans on $\mathcal{G}$, where $[\mathcal{C}]^k=\underbrace{\mathcal{C}\times\cdots\times\mathcal{C}}_{k}$.
\end{theorem}
  \vspace{-0.17in}
To prove Theorem \ref{the-constant}, we first introduce two lemmas.  %In the our following discussion, we use $<a, b>$ to denote the inner product of $a$ and $b$. 
  \vspace{-0.05in}
  \begin{lemma}
\label{lem-meanshift}
Let $P$ be a set of points in $\mathbb{R}^d$ space, and $m$ be the mean point of $P$. For any point $m'\in \mathbb{R}^d$, $\sum_{p\in P}||p-m'||^2=\sum_{p\in P}||p-m||^2+|P| \times ||m-m'||^2$ (see Appendix for the proof).

\end{lemma}
  \vspace{-0.05in}
%Due to space limit, we put the proof for Lemma \ref{lem-meanshift}  in Appendix (Section \ref{sec-meanshift}).
  \vspace{-0.06in}
  \begin{lemma}
\label{lem-close}
Let $P$ be a set of points in $\mathbb{R}^d$ space, and $P_{1}$ be its subset containing $\alpha |P|$ points for some $0<\alpha\leq 1$. Let $m$ and $m_1$ be the mean points of $P$ and $P_{1}$ respectively. Then  $||m_{1}-m||\leq\sqrt{\frac{1-\alpha}{\alpha}}\delta$, where $\delta^2=\frac{1}{|P|}\sum_{p\in P}||p-m||^2$.
\end{lemma}
\begin{proof}
Let $P_2=P\setminus P_1$, and $m_{2}$ be its mean point. By Lemma \ref{lem-meanshift} we first have the following two equalities.
 \vspace{-0.07in}
 \small{
  \begin{eqnarray} 
    \sum_{p\in P_{1}}||p-m||^2 &= &\sum_{p\in P_{1}}||p-m_1||^2+|P_1| \times ||m_1-m||^2. \label{for-1}\\
% $$=\sum_{p\in P_{1}}(||p-m_1||^2+2<p-m_1, m_1-m>+||m_1-m||^2)$$
% $$=\sum_{p\in P_{1}}||p-m_1||^2+2\sum_{p\in P_{1}}<p-m_1, m_1-m>+|P_1|*||m_1-m||^2$$
% $$=\sum_{p\in P_{1}}||p-m_1||^2+|P_1|*||m_1-m||^2.$$
% The last equality is from that $\sum_{p\in P_{1}}<p-m_1, m_1-m>=<\sum_{p\in P_{1}}(p-m_1), m_1-m>=0$ (since $m_1$ is the mean point of $P_1$).
% 
% Similarly, we have
  \sum_{p\in P_{2}}||p-m||^2 &=& \sum_{p\in P_{2}}||p-m_2||^2+|P_2| \times ||m_2-m||^2. \label{for-2}
 \end{eqnarray} 
 }
 \normalsize
\vspace{-0.05in} 

Then by the definition of $\delta$, we have $\delta^2=\frac{1}{|P|}(\sum_{p\in P_{1}}||p-m||^2+\sum_{p\in P_{2}}||p-m||^2)$. Let $L=||m_{1}-m_{2}||$. By the definition of mean point, we have $m=\frac{1}{|P|}\sum_{p\in P} p=\frac{1}{|P|}(\sum_{p\in P_1} p+\sum_{p\in P_2} p)=\frac{1}{|P|}(|P_1|m_1+|P_2|m_2)$. Thus the three points $\{m, m_1, m_2\}$ are collinear, and $||m_{1}-m||=(1-\alpha) L$ and $||m_{2}-m||=\alpha L$. Combining (\ref{for-1}) and (\ref{for-2}), we have 

\vspace{-0.05in}
\small{
\begin{eqnarray*}
\delta^2= \frac{1}{|P|}(\sum_{p\in P_{1}}||p-m_1||^2+ |P_1| \times ||m_1-m||^2+\sum_{p\in P_{2}}||p-m_2||^2+|P_2| \times ||m_2-m||^2)\\
\geq \frac{1}{|P|}( |P_1| \times ||m_1-m||^2+|P_2| \times ||m_2-m||^2)
 =\alpha((1-\alpha)L)^2+(1-\alpha)(\alpha L)^2=\alpha(1-\alpha)L.
 \end{eqnarray*} 
 }
 \normalsize 
 \vspace{-0.07in}  
Thus, we have $L\leq\frac{\delta}{\sqrt{\alpha(1-\alpha)}}$, which means that $||m_{1}-m||=(1-\alpha)
L\leq\sqrt{\frac{1-\alpha}{\alpha}}\delta$.
\qed
\end{proof}
  \vspace{-0.03in}
\begin{proof}[\textbf{of Theorem \ref{the-constant}}]
Let $\{c_1, \cdots, c_k\}$ be the $k$ mean points in  $\mathcal{C}$, and $\{S_1, \cdots, S_k\}$ be their corresponding clusters. 
%Thus, $c_j$ is the mean point of $S_j$ for $1\leq j\leq k$. 
Let $\{m_{1}, \cdots, m_{k}\}$ be the $k$ unknown optimal mean points of $k$-CMeans, and $\mathcal{OPT}=\{Opt_{1}, \cdots, Opt_{k}\}$ be the corresponding $k$ optimal chromatic clusters.
%, we denote the \textbf{unknown} optimal solution for $k$-CMeans on $\mathcal{G}$ as  $\mathcal{OPT}=\{Opt_{1}, \cdots, Opt_{k}\}$, and the mean point of each $Opt_j$ as $m_j$ for $1\leq j\leq k$. 
Let $\Gamma^i_j=Opt_i\cap S_j$, and $\tau^i_j$ be its mean point for $1\leq i, j\leq k$.
\vspace{-0.18in}
\begin{figure}[ht]
\vspace{-0.1in}
  \centerline{
  \includegraphics[height=1.15in]{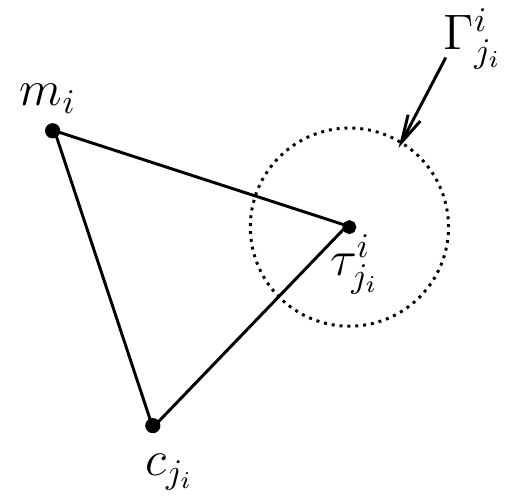}}
  \vspace{-0.15in}
    \caption{An example illustrating Theorem \ref{the-constant}.}
  \label{fig-constant}
  \vspace{-0.2in}
\end{figure}
\vspace{-0.05in}
Since $\cup^k_{j=1}\Gamma^i_j=Opt_i$, by pigeonhole principle we know that  there must exist some index $1\leq j_i \leq k$ such that $|\Gamma^i_{j_i}|\geq \frac{1}{k}|Opt_i|$. Thus by fixing $j_i$, we have the following about $\sum_{p\in Opt_i}||p-c_{j_i}||^2$ (see Figure \ref{fig-constant})
%  and consider the value $\sum_{p\in Opt_i}||p-c_{j_i}||^2$ (See Figure. \ref{fig-constant}): 
 \small{
 \begin{eqnarray}
 \sum_{p\in Opt_i}||p-c_{j_i}||^2 &= &\sum_{p\in Opt_i}||p-m_i||^2+|Opt_i| \times ||m_i-c_{j_i}||^2 %\nonumber\\
=\sum_{p\in Opt_i}||p-m_i||^2+|Opt_i| \times ||m_i-\tau^i_{j_i}+\tau^i_{j_i}-c_{j_i}||^2 \nonumber\\
&\leq&\sum_{p\in Opt_i}||p-m_i||^2+|Opt_i| \times (||m_i-\tau^i_{j_i}||+||\tau^i_{j_i}-c_{j_i}||)^2 \nonumber\\
&\leq&\sum_{p\in Opt_i}||p-m_i||^2+|Opt_i| \times 2(||m_i-\tau^i_{j_i}||^2+||\tau^i_{j_i}-c_{j_i}||^2), \label{for-3}
\end{eqnarray}
}
\normalsize
where the first equation follows from Lemma \ref{lem-meanshift} (note that $m_i$ is the mean point of $Opt_i$), and the last inequality follows from the fact that $(a+b)^2\leq 2(a^2+b^2)$ for any numbers $a$ and $b$. By Lemma \ref{lem-close}, we have
\vspace{-0.12in}
\small{
\begin{eqnarray}
  ||\tau^i_{j_i}-m_i||^2\leq \frac{1-\frac{1}{k}}{\frac{1}{k}}(\frac{1}{|Opt_i|}\sum_{p\in Opt_i}||p-m_i ||^2). \label{for-4}\\
  ||\tau^i_{j_i}-c_{j_i}||^2\leq \frac{1-\frac{|\Gamma^i_{j_i}|}{|S_{j_i}|}}{\frac{|\Gamma^i_{j_i}|}{|S_{j_i}|}}(\frac{1}{|S_{j_i}|}\sum_{p\in S_{j_i}}||p-c_{j_i} ||^2). \label{for-5}
  \end{eqnarray}
  }
  \normalsize
%Meanwhile, also from Lemma \ref{lem-close}, we have:
%$$\textbf{(5).}||\tau^i_{j_i}-c_{j_i}||^2\leq \frac{1-\frac{|\Gamma^i_{j_i}|}{|S_{j_i}|}}{\frac{|\Gamma^i_{j_i}|}{|S_{j_i}|}}(\frac{1}{|S_{j_i}|}\sum_{p\in S_{j_i}}||p-c_{j_i} ||^2).$$
Plugging (\ref{for-4}) and (\ref{for-5}) into inequality (\ref{for-3}), we have
\vspace{-0.05in}
\small{
\begin{eqnarray*}
\sum_{p\in Opt_i}||p-c_{j_i}||^2\leq\sum_{p\in Opt_i}||p-m_i||^2+|Opt_i| \times 2(||m_i-\tau^i_{j_i}||^2+||\tau^i_{j_i}-c_{j_i}||^2)\\
\vspace{-0.05in}
\leq\sum_{p\in Opt_i}||p-m_i||^2+|Opt_i| \times 2(\frac{1-\frac{1}{k}}{\frac{1}{k}}(\frac{1}{|Opt_i|}\sum_{p\in Opt_i}||p-m_i ||^2)+\frac{1-\frac{|\Gamma^i_{j_i}|}{|S_{j_i}|}}{\frac{|\Gamma^i_{j_i}|}{|S_{j_i}|}}(\frac{1}{|S_{j_i}|}\sum_{p\in S_{j_i}}||p-c_{j_i} ||^2))\\
\vspace{-0.05in}
=(2k-1)\sum_{p\in Opt_i}||p-m_i||^2+2\frac{|Opt_i|}{|\Gamma^i_{j_i}|} \times (1-\frac{|\Gamma^i_{j_i}|}{|S_{j_i}|})\sum_{p\in S_{j_i}}||p-c_{j_i} ||^2).
\end{eqnarray*}
}
\normalsize 
\vspace{-0.05in}
Since $|\Gamma^i_{j_i}|\geq \frac{1}{k}|Opt_i|$, we have $\frac{|Opt_i|}{|\Gamma^i_{j_i}|} \times (1-\frac{|\Gamma^i_{j_i}|}{|S_{j_i}|})\leq k$. Thus the above inequality becomes
\small{
\begin{eqnarray}
\sum_{p\in Opt_i}||p-c_{j_i}||^2\leq (2k-1)\sum_{p\in Opt_i}||p-m_i||^2+2k\sum_{p\in S_{j_i}}||p-c_{j_i} ||^2. \label{for-6}
\end{eqnarray}
}
\normalsize
Summing both sides of (\ref{for-6}) over $i$, we have
\small{
\begin{eqnarray}
\sum^k_{i=1}\sum_{p\in Opt_i}||p-c_{j_i}||^2\leq (2k-1)\sum^k_{i=1}\sum_{p\in Opt_i}||p-m_i||^2+2k\sum^k_{i=1}\sum_{p\in S_{j_i}}||p-c_{j_i} ||^2 \nonumber\\
\leq (2k-1)\sum^k_{i=1}\sum_{p\in Opt_i}||p-m_i||^2+2k^2 \sum^k_{j=1}\sum_{p\in S_{j}}||p-c_{j} ||^2, \label{for-7}
\end{eqnarray}
}
\normalsize
where the second inequality follows from the inequality $\sum_{p\in S_{j_i}}||p-c_{j_i} ||^2\leq \sum^k_{j=1}\sum_{p\in S_{j}}||p-c_{j} ||^2$, which implies that $2k\sum^k_{i=1}\sum_{p\in S_{j_i}}||p-c_{j_i} ||^2\leq 2k^2 \sum^k_{j=1}\sum_{p\in S_{j}}||p-c_{j} ||^2$.

It is obvious that the optimal objective value of $k$-means is no larger than that of $k$-CMeans on the same set of points in $\mathcal{G}$.  Thus, $\sum^k_{j=1}\sum_{p\in S_{j}}||p-c_{j} ||^2\leq c\sum^k_{i=1}\sum_{p\in Opt_i}||p-m_i||^2$. Plugging this inequality into inequality (\ref{for-7}), we have
\small{
\begin{eqnarray*}
\sum^k_{i=1}\sum_{p\in Opt_i}||p-c_{j_i}||^2\leq (2ck^2+2k-1)\sum^k_{i=1}\sum_{p\in Opt_i}||p-m_i||^2.
\end{eqnarray*}
}
\normalsize
 The above inequality means that if we take the $k$-tuple $(c_{j_1}, \cdots, c_{j_k})$ as the $k$ approximate mean points for $k$-CMeans, we have a $(2ck^2+2k-1)$-approximation solution, where the $k$ chromatic clusters can be obtained by the bipartite matching procedure. Thus, the theorem is proved.
%and compute the minimum weight bipartite matching for each $G_i$ to $(c_{j_1}, \cdots, c_{j_k})$, we would get $(2ck^2+2k-1)$-approximation for $k$-CMeans on $\mathcal{G}$. So the theorem is proved. 
\qed
\end{proof}
  \vspace{-0.1in}
  
  \noindent \textbf{Running Time:} In the above theorem, the bipartite matching procedure takes $O(k^3 nd)$ time for one $k$-tuple.  Since there are in total $O(k^k)$ such $k$-tuples,
%from $[\mathcal{C}]^k$. 
the total running time is $O(k^{k+3} nd)$ for computing a $(2ck^2+2k-1)$-approximation of $k$-CMeans from a $c$-approximation of $k$-means. As $k$ is assumed to be a constant in this paper, the running time is linear.
\vspace{-0.2in}
\section{$(1+\epsilon)$-Approximation Algorithm}
\label{sec-ptas}
\vspace{-0.15in}

This section presents our $(1+\epsilon)$-approximation solution to the $k$-CMeans problem. We first introduce a standalone result, {\em Simplex Lemma}, and then use it to achieve a $(1+\epsilon)$-approximation for $k$-CMeans. The main idea of the algorithm is to use a sphere peeling technique to generate the chromatic clusters iteratively, where the Simplex Lemma helps to determine a proper peeling region. 

%With the above theorem, we now focus on designing approximation solutions for the $k$-CMeans problem.
  \vspace{-0.15in}
\subsection{Simplex Lemma}
\label{sec-simplex}
  \vspace{-0.05in}

Simplex Lemma is mainly for approximating the mean point of some \textbf{unknown} points set $P$. The only known information about $P$ is a set $S$ of $j$ points with each of them being an approximate mean 
point of a subset of $P$.  The following Simplex lemmas show that it is possible to construct a simplex of $S$ and find the desired approximate mean point of $P$ inside the simplex.

\begin{figure}[]
\vspace{-0.25in}
\begin{minipage}[t]{0.5\linewidth}
  \centering
  \includegraphics[height=1.0in]{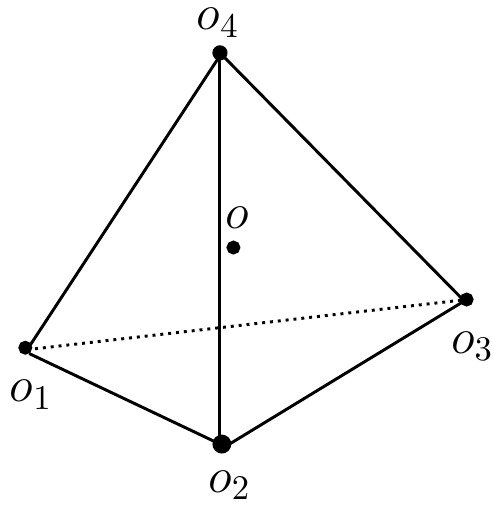}
  \vspace{-0.15in}
     \caption{An example for Lemma \ref{lem-simplex} with $j=4$.}
  \label{fig-simplex}
\end{minipage}
\begin{minipage}[t]{0.5\linewidth}
\centering
  \includegraphics[height=1.1in]{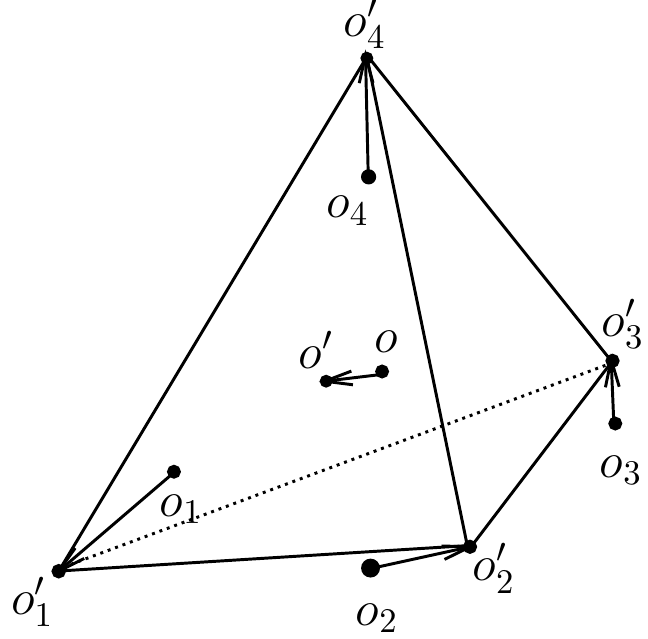}
  \vspace{-0.15in}
       \caption{An example for Lemma \ref{lem-shift} with $j=4$.}
  \label{fig-simplex2}
  \end{minipage}
  \vspace{-0.28in}
\end{figure}
\vspace{-0.1in}
\begin{lemma}[Simplex Lemma \Rmnum{1}]
\label{lem-simplex}
Let $P$  be a set of points in $\mathbb{R}^d$ with a partition of $P=\cup^j_{l=1} P_l$ and $P_{l_1}\cap P_{l_2}=\emptyset$ for any $l_1\neq l_2$. 
Let $o$ be the mean point of $P$, and $o_l$ be the mean point of $P_l$ for $1\leq l\leq j$. Further, let $\delta^2=\frac{1}{|P|}\sum_{p\in P}||p-o||^2$,  and $V$ be the simplex determined by $\{o_1, \cdots, o_j\}$. 
Then for any $0<\epsilon\leq 1$, it is possible to construct a grid of size $O((8j/\epsilon)^j)$ inside $V$ such that at least one grid point $\tau$ satisfies the inequality $||\tau-o||\leq\sqrt{\epsilon}\delta$.
%If $\{m'_{1}, \cdots, m'_{j+1}\}$ satisfies the conditions of $||m'_{l}-m_{l}||^2\leq\epsilon\delta^2$ for $1\leq l\leq j$ and $||m'_{j+1}-m_{j+1}||^2\leq\epsilon\delta^{2}_{j+1}$,
%there exists an $O(d(2j/\epsilon)^j)$-time algorithm to generate an approximate mean point $m'$ of $Q$ in $\mathbb{R}^d$ such that $||m'-m||^2\leq O(\epsilon)\delta^2$. 
%The running time is $O(d(2j/\epsilon)^j)$.
\end{lemma}
%\vspace{-0.05in}
\begin{proof}
We will prove this lemma by mathematical induction on $j$. 

\noindent\textbf{Base case:} For $j=1$, since $P_1=P$, $o_1=o$. Thus, the simplex $V$ and the grid are all simply the point $o_1$. Clearly $\tau=o_{1}$ satisfies the inequality. 
%and the grid is trivially the point $o_{1} to get the desired $\tau$ (just assign $o_1=\tau$).

\noindent\textbf{Induction step:} Assume that the lemma holds for any $j\leq j_0$ for some $j_{0} \ge 1$ (i.e., Induction Hypothesis). Now we consider  the case of $j=j_0+1$. First, we assume that $\frac{|P_l |}{|P|}\geq \frac{\epsilon}{4j}$ for each $1\leq l\leq j$. Otherwise, we can reduce the problem to the case of smaller $j$ in the following way. Let $I=\{l| 1\leq l\leq j, \frac{|P_l |}{|P|}< \frac{\epsilon}{4j}\}$ be the index set of small subsets. Then, $\frac{\sum_{l\in I}|P_l |}{|P|}<\frac{\epsilon}{4}$, and $\frac{\sum_{l\not\in I}|P_l |}{|P|}\geq 1-\frac{\epsilon}{4}$. By Lemma \ref{lem-close}, we know that  $||o'-o||\leq\sqrt{\frac{\epsilon/4}{1-\epsilon/4}}\delta$, where $o'$ is the mean point of $\cup_{l\not\in I}P_l$.  Let $(\delta')^2$ be the variance of $\cup_{l\not\in I}P_l$. Then, we have $(\delta')^2\le\frac{|P|}{|\cup_{l\not\in I}P_l|}\delta^2\leq \frac{1}{1-\epsilon/4}\delta^2$. Thus, if we replace $P$ and $\epsilon$ by $\cup_{l\not\in I}P_l$ and $\frac{\epsilon}{16}$ respectively, and find a point $\tau$ such that $||\tau-o'||^2\leq\frac{\epsilon}{16}(\delta')^2\leq \frac{\epsilon/16}{1-\epsilon/4}\delta^2$, we have $||\tau-o||^2\leq(||\tau-o'||+||o'-o||)^2\leq \frac{\frac{9}{16}\epsilon}{1-\epsilon/4}\delta^2\leq \epsilon\delta^2$ (where the last inequality is due to the fact $\epsilon<1$). This means that we can reduce the problem to a problem with  point set $\cup_{l\not\in I}P_l$ and a smaller $j$ (i.e., $j-|I|$). By the induction hypothesis, we know that the reduced problem can be solved (note that the simplex would be a subset of $V$ determined by $\{o_l\mid 1\leq l\leq j, l\not\in I\}$), and therefore the induction step holds for this case.  Thus, in the following discussion, we can assume that $\frac{|P_l |}{|P|}\geq \frac{\epsilon}{4j}$ for each $1\leq l\leq j$.

For each $1\leq l\leq j$, since $\frac{|P_l |}{|P|}\geq \frac{\epsilon}{4j}$, by Lemma \ref{lem-close}, we know that $||o_l-o||\leq\sqrt{\frac{1- \frac{\epsilon}{4j}}{ \frac{\epsilon}{4j}}}\delta\leq 2\sqrt{\frac{j}{\epsilon}}\delta$. This, together with triangle inequality, implies that  for any $1\leq l, l'\leq j$, $||o_l-o_{l'}||\leq ||o_l-o||+||o_{l'}-o||\leq 4\sqrt{\frac{j}{\epsilon}}\delta$. Thus, if we pick any index $l_0$, and draw a ball $\mathcal{B}$ centered at $o_{l_0}$ and with radius $r=\max_{1\leq l\leq j}\{||o_l-o_{l_0}||\}\leq 4\sqrt{\frac{j}{\epsilon}}\delta$,  the whole simplex $V$ will be inside $\mathcal{B}$. Note that since $o=\sum^j_{l=1}\frac{|P_j|}{|P|}o_l$, $o$ also locates inside $V$. This indicates that we can construct $\mathcal{B}$ in the $j-1$-dimensional space spanned by $\{o_1, \cdots, o_j\}$, rather than the whole $\mathbb{R}^d$ space.  Also, if we build a grid inside $\mathcal{B}$ with grid length $\frac{\epsilon r}{4j}$,  the total number of  grid points is no more than $O((\frac{8j}{\epsilon})^j)$. With this grid, we know that for any point $q$ inside $V$, there exists a grid point $g$ such that $||g-q||\leq \sqrt{j (\frac{\epsilon r}{4j})^2}=\frac{\epsilon}{4\sqrt{j}}r\leq \sqrt{\epsilon}\delta$. This means that can find a grid point $\tau$ inside $V$, such that $||\tau-o||^2\leq\epsilon\delta^2$. Thus, the induction step holds.

With the above base case and  induction steps, the lemma holds for any $j\ge 1$.
% In conclusion, the above mathematics induction holds, and the lemma is proved.
\qed
\end{proof}
  \vspace{-0.05in}
In the above lemma, we assume that the exact positions of $\{o_1, \cdots, o_j\}$ are known (see Fig. \ref{fig-simplex}). However, in some scenario (e.g., the exact partition of $P$ is not given, as is the case in $k$-CMeans), 
%the chromatic clustering problem), 
it is possible that we only know the approximate position of each mean point $o_{i}$ (see Fig. \ref{fig-simplex2}). The following lemma shows that an approximate position of $o$ can still be similarly determined.

 %some situation, we do not know the precise positions of  $\{o_1, \cdots, o_j\}$. Instead, we just have the approximate positions of them, then the following lemma shows we can also get the approximate position of $o$.

  \vspace{-0.07in}
\begin{lemma}[Simplex Lemma \Rmnum{2}]
\label{lem-shift}
Let $P$, $o$, $P_{l}, o_{l}, 1\le l \le j$, and, $\delta$ be defined as in Lemma \ref{lem-simplex}. 
%Using the same symbols in Lemma \ref{lem-simplex}, 
Let  $\{o'_1, \cdots, o'_j\}$ be  $j$ points in $R^{d}$ such that $||o'_l-o_l ||\leq L$ for $1\leq l\leq j$ and $L>0$, and $V'$ be the simplex determined by $\{o'_1, \cdots, o'_j\}$. Then for any $0<\epsilon\leq 1$, it is possible to construct a grid of size $O((8j/\epsilon)^j)$ inside $V'$ such that at least one grid point $\tau$ satisfies the inequality $||\tau-o||\leq\sqrt{\epsilon}\delta+(1+\epsilon)L$.% (see Appendix for the proof).
% can construct a grid with size $O((4j/\epsilon)^j)$ inside $V'$, which contains one grid point $\tau$ such that $||\tau-o||\leq\sqrt{\epsilon}\delta+(1+\epsilon)L$.
\end{lemma}
\vspace{-0.07in}
  \vspace{-0.25in}
\subsection{Sphere Peeling Algorithm}
\label{sec-peeling}
  \vspace{-0.1in}

This section presents a sphere peeling algorithm to achieve a $(1+\epsilon)$-approximation for $k$-CMeans. 

Let $\mathcal{G}=\{G_1, \cdots, G_n\}$ be an instance of $k$-CMeans with $k$ (unknown) optimal chromatic clusters $\mathcal{OPT}=\{Opt_1, \cdots,$ $Opt_k\}$, and  $m_{j}$ be the mean point of the cluster $Opt_j$ for $1\leq j\leq k$. 
%Note that $\mathcal{OPT}$ and $\{m_1, \cdots, m_k\}$ are unknown. 
%Actually, the main task for $k$-CMeans is to find $\{m_1, \cdots, m_k\}$, since after fix these $k$ points, we only need to use bipartite matching algorithm for each $G_i$ and then get the final solution. 
Without loss of generality, we assume that $|Opt_1|\geq |Opt_2|\geq\cdots\geq |Opt_k |$.  

\noindent\textbf{Algorithm overview:} Our algorithm first computes a constant  $C$-approximation solution (by Theorem \ref{the-constant}) to determine an upper bound  $\Delta$ of the optimal objective value $\delta^{2}_{opt}$, and then search for a good approximation of $\delta^{2}_{opt}$ in the interval of $[\Delta/C, \Delta]$. At each search step, our algorithm performs a sphere peeling procedure to iteratively generate $k$ approximate mean points for the chromatic clusters.    Initially, the sphere peeling procedure uses random sampling technique (i.e., Lemma \ref{lem-dis} and \ref{lem-select}) to find an approximate mean point for $Opt_{1}$. At $(j+1)$-th iteration,  it already has approximate mean  points $\{p_{v_1}, \cdots, p_{v_j}\}$ for $Opt_1, \cdots, Opt_j$ respectively. Then it draws $j$ peeling spheres, $B_{j+1,1}, \cdots, B_{j+1,j}$, centered at the $j$ approximate mean points respectively and with  a radius determined by  the approximation of $\delta_{opt}$.
%constant approximation solution from Theorem \ref{the-constant}. 
Denote the set of unknown points $Opt_{j+1}\setminus (\cup^j_{l=1}B_{j+1,l})$ as $\mathcal{A}$. Our algorithm considers two cases: (a) $|\mathcal{A}|$ is large enough and (b) $|\mathcal{A}|$ is small. For case (a), since $|\mathcal{A}|$ is large enough, we can first use Lemma \ref{lem-select} to find an approximate mean point $m_{\mathcal{A}}$ of $\mathcal{A}$, and then construct a simplex determined by $m_{\mathcal{A}}$ and $\{p_{v_1}, \cdots, p_{v_j}\}$. For case (b), it directly constructs a simplex determined just by $\{p_{v_1}, \cdots, p_{v_j}\}$. For either  case, our algorithm builds a grid inside the simplex (i.e., using Lemma \ref{lem-shift}) to find an approximate mean point for $Opt_{j+1}$ (i.e., $p_{v_{j+1}}$). Repeat the sphere peeling procedure $k$ times to generate the $k$ approximate mean points.

%where $c_{1}, \cdots, c_{j-1}$ are the approximate mean points generated in the previous $j-1$ iterations. Consider two groups of points, (a) points in $Opt_{j}$ but not in the union of the $j-1$ peeling spheres  and (b) those points peeled (or covered) by the $j-1$ spheres (from $Opt_{j}$). Depending on the size of group (a) and using the Lemma \ref{lem-shift}, either determine a new approximate mean point $c_{j}$ for $Opt_{j}$ or use one of the existing approximate mean points $c_{1}, \cdots, c_{j-1}$ as the approximate mean point $c_{j}$ for $Opt_{j}$. 

%\noindent\textbf{Note} that as mentioned earlier, we can not directly use the $k$-means peeling algorithm in \cite{KSS}, due to the lack of locality property in $k$-CMeans.
% and yields $(1+\epsilon)$-approximation, since the idea from \cite{KSS} does not consider the chromatic requirement. Although it would output $(1+\epsilon)$-approximation for $k$-meanss, it can not guarantee $(1+\epsilon)$-approximation for $k$-CMeans.\\

\vspace{0.05in}

\noindent\textbf{Algorithm $k$-CMeans}
\newline \textbf{Input:} $\mathcal{G}=\{G_1, \cdots, G_n\}$, $k\geq 2$, and a small positive value $\epsilon$.
\newline \textbf{Output:} $(1+\epsilon)$-approximation solution for $k$-CMeans on $\mathcal{G}$.
\vspace{-0.1in}
\begin{enumerate}
\item Run the PTAS
 %$(1+\epsilon)$-approximation algorithm 
 of $k$-means in \cite{KSS} on $\mathcal{G}$, and let $\Delta$ be the obtained objective value.

\item For $i=1$ to $\frac{2k}{\epsilon}$ do
\begin{enumerate}
\item Set $\delta=\frac{\sqrt{\Delta}}{2k}+i\frac{\epsilon}{2k}\sqrt{\Delta}$, and run the Sphere-Peeling-Tree algorithm.

\item  Let $\mathcal{T}_i$ be the output tree.

\end{enumerate}

\item For each path of every $\mathcal{T}_i$, use bipartite matching procedure to compute the objective value of $k$-CMeans on $\mathcal{G}$. Output the $k$ points from the path with the smallest objective value.

\end{enumerate}

%The following {\em Peeling Tree Algorithm} is called by Algorithm $1$.\\

\vspace{-0.05in}
\noindent\textbf{Algorithm Sphere-Peeling-Tree}
\newline \textbf{Input:} $\mathcal{G}$, $k\geq 2$, $\epsilon, \delta>0$.
\newline \textbf{Output:} A tree $\mathcal{T}$ of height $k$ with each node $v$ associating with a point $p_v \in \mathbb{R}^d$.
\vspace{-0.1in}
\begin{enumerate}
%\item Denote the radius candidates set $\mathcal{R}=\bigcup^{6+\frac{2}{\epsilon}}_{l=1}\{2^t\frac{1+(1+l)\frac{\epsilon}{2}}{2(1+\epsilon)}j\sqrt{\epsilon}\delta\mid t=0, 1, \cdots, \frac{\log(kn)}{2}\}$.
\item Initialize $\mathcal{T}$ with a single root node $v$ associating with no point.
\item Recursively grow each node $v$ in the following way
\begin{enumerate}
\item If the height of $v$ is already $k$, then it is a leaf.
\item Otherwise, let $j$ be the height of $v$. Build the radius candidates set $\mathcal{R}=\cup^{\log(kn)}_{t=0}\{\frac{1+l\frac{\epsilon}{2}}{2(1+\epsilon)}j2^{t/2}\sqrt{\epsilon}\delta\mid 0\le l\le 4+\frac{2}{\epsilon}\}$. For each $r\in\mathcal{R}$, do
\begin{enumerate}
\item   Let $\{p_{v_1}, \cdots, p_{v_j}\}$ be the $j$ points associated with nodes on the root-to-$v$ path. % (including $p_v$) . 

\item For each $p_{v_l}$, $1\leq l\leq j$, construct a ball $B_{j+1,l}$ centered at $p_{v_l}$ and with radius $r$. 
\item Take a random sample from $\mathcal{G}\setminus\cup^j_{l=1}B_{j+1,l}$ with size $m=\frac{8k^3}{\epsilon^9}\ln\frac{k^2}{\epsilon^6}$. Compute the mean points of all subset of the sample, and denote them as $\Pi=\{\pi_1, \cdots, \pi_{2^m-1}\}$.
\item For each $\pi_i \in \Pi$, construct the simplex determined by $\{p_{v_1}, \cdots, p_{v_j}, \pi_i\}$. Also construct the simplex determined by $\{p_{v_1}, \cdots, p_{v_j}\}$. Build a grid inside each simplex with size $O((\frac{32j}{\epsilon^2})^j)$. %as Theorem \ref{the-constant}. 

\item In total, there are $2^m (\frac{32j}{\epsilon^2})^j$ grid points inside the $2^m$ simplices. For each grid point, add one child to $v$, and associate it with the grid point.

\end{enumerate}

\end{enumerate}

\end{enumerate}

\vspace{-0.1in}
\begin{theorem}
\label{the-ptas}
With constant probability, Algorithm $k$-CMeans yields a $(1+\epsilon)$-approximation for $k$-CMeans in $O(2^{poly(\frac{k}{\epsilon})}n(\log n)^{k+1} d )$ time. %where $f$ is a polynomial function of $\frac{1}{\epsilon}, k$.
\end{theorem}

\vspace{-0.3in}
\subsection{Proof of Theorem \ref{the-ptas}}
\label{sec-proof}
\vspace{-0.07in}

Let  $\beta_j= |Opt_j|/|\cup_{i=1}^{n} G_{i}|$, % (i.e., $\beta_{j}$ is the faction of points in the cluster of $Opt_{j}$), 
and 
%each $1\leq j\leq k$, we denote the fractions for $Opt_j$ respect to $\mathcal{G}$ as $\beta_j$, and 
 $\delta^2_j=\frac{1}{|Opt_j |}\sum_{p\in Opt_j}||p-m_j||^2$, where $m_j$ is the mean point of $Opt_j$. 
 Clearly, $\beta_1\geq\cdots\geq\beta_k$ (by assumption) and $\sum^k_{j=1}\beta_j =1$. Let $\delta^2_{opt}=\sum^k_{j=1}\beta_j\delta^2_j$. 
 
We prove Theorem \ref{the-ptas} by mathematical induction.  Instead of directly proving it, we consider the following two lemmas which jointly ensure the correctness of Theorem \ref{the-ptas}.
\vspace{-0.05in}
\begin{lemma}
\label{lem-induction}
Among all the trees generated in Algorithm $k$-CMeans, with constant probability, there exists at least one tree, $\mathcal{T}_i$, which has a root-to-leaf path with each node $v_{j}$ at level $j$, $1\leq j\leq k$, on the path  associating a point $p_{v_j}$ and  satisfying the inequality
%\vspace{-0.1in}
$||p_{v_j}-m_j ||\leq \epsilon\delta_j+(1+\epsilon)j\sqrt{\frac{\epsilon}{\beta_j}}\delta_{opt} .$
%\vspace{-0.1in}
\end{lemma}
\vspace{-0.09in}
Before proving this lemma, we first show its implication. 
%the following lemma shows its relation with Theorem \ref{the-ptas}:
\vspace{-0.08in}
\begin{lemma}
\label{lem-equal}
If Lemma \ref{lem-induction} is true, Algorithm $k$-CMeans yields a $(1+O(k^3)\epsilon)$-approximation for $k$-CMeans.
\end{lemma}
\begin{proof}
We first assume that Lemma \ref{lem-induction} is true. Then for each $1\leq j\leq k$, we have
\small{
\begin{eqnarray}
\sum_{p\in Opt_j}||p-p_{v_j}||^2&=&\sum_{p\in Opt_j}||p-m_j||^2+|Opt_j|\times||m_j-p_{v_j}||^2 %\nonumber\\
\leq \sum_{p\in Opt_j}||p-m_j||^2+|Opt_j|\times2(\epsilon^2\delta^2_j+(1+\epsilon)^2 j^2\frac{\epsilon}{\beta_j}\delta^2_{opt}) \nonumber\\
&=&(1+2\epsilon^2)|Opt_j |\delta^2_j+2(1+\epsilon)^2 j^2\epsilon|\mathcal{G}|\delta^2_{opt}, \label{for-10}
\end{eqnarray}
}
\normalsize
where the first equation follows from Lemma \ref{lem-meanshift} (note that $m_j$ is the mean point of $Opt_j$), the second inequality follows from Lemma \ref{lem-induction} and the fact that $(a+b)^2\leq 2(a^2+b^2)$ for any two real numbers $a$ and $b$, and the last equality follows from $\frac{|Opt_j|}{\beta_j}=|\mathcal{G}|$. Summing both sides of (\ref{for-10}) over $j$, we have
\small{
\begin{eqnarray}
\sum^k_{j=1}\sum_{p\in Opt_j}||p-p_{v_j}||^2 &\leq &\sum^k_{j=1}((1+2\epsilon^2)|Opt_j |\delta^2_j+2(1+\epsilon)^2 j^2\epsilon|\mathcal{G}|\delta^2_{opt})\nonumber\\
&\leq& (1+2\epsilon^2)\sum^k_{j=1}|Opt_j |\delta^2_j+2(1+\epsilon)^2 k^3\epsilon|\mathcal{G}|\delta^2_{opt}=(1+O(k^3)\epsilon)|\mathcal{G}|\delta^2_{opt}, \label{for-11}
\end{eqnarray}
}
\normalsize
%\vspace{-0.05in}
where the last equation follows from the fact that $\sum^k_{j=1}|Opt_j |\delta^2_j=|\mathcal{G}|\delta^2_{opt}$. By (\ref{for-11}), we know that $\{p_{v_1}, \cdots, p_{v_k}\}$ will induce a $(1+O(k^3)\epsilon)$-approximation solution for $k$-CMeans via bipartite matching procedure. Since Algorithm $k$-CMeans outputs the best solution generated in all trees, the resulting solution is clearly a  $(1+O(k^3)\epsilon)$-approximation solution. Thus the lemma is true.
\qed
\end{proof}
\vspace{-0.06in}
The above lemma indicates that  if we replace $\epsilon$ by $\frac{\epsilon}{k^3}$  in the input of  our algorithm, it will result in a $(1+\epsilon)$-approximation solution. This implies that Lemma \ref{lem-induction} is indeed sufficient to ensure the correctness of Theorem \ref{the-ptas} (except for the time complexity). Now we prove Lemma \ref{lem-induction}.
%by Lemma \ref{lem-equal}. 
%For simplicity of our presentation, we just use $\epsilon$, rather than $\frac{\epsilon}{k^3}$, in the following analysis.

%Now we discuss our ideas for proving Lemma \ref{lem-induction}.  We first establish the following claim.

%\begin{claim}[\textbf{1}]
%For any $1\leq j\leq k$, $\delta^2_j\leq \frac{1}{\beta_j}\delta^2_{opt}$.
%\end{claim}
%\begin{proof}
%It directly follows from the fact that $\delta^2_{opt}=\sum^k_{j=1}\beta_j\delta^2_j$.
%\qed
%\end{proof}

%\vspace{-0.08in}

\vspace{-0.08in}
\begin{proof}[\textbf{of Lemma \ref{lem-induction}}]
Note that $\Delta\le4k^2\delta^2_{opt}$, and we build $\epsilon$-net in $[\frac{\sqrt{\Delta}}{2k},\sqrt{\Delta}]$. Let $\mathcal{T}_i$ be the tree generated by Algorithm Sphere-Peeling-Tree and corresponding to the input $\delta\in [\delta_{opt}, (1+\epsilon)\delta_{opt}]$.  We will focus our discussion on $\mathcal{T}_{i}$, and prove the lemma by mathematical induction on $j$.

\noindent\textbf{Base case:} For $j=1$, since $\beta_1=\max\{\beta_j |1\leq j\leq k\}$, we have $\beta_1\geq\frac{1}{k}$. By Lemmas \ref{lem-dis} and \ref{lem-select}, we can find the approximation mean point through random sampling. Let $p_{v_{1}}$ be the approximation mean point.  Clearly, $||p_{v_1}-m_1||\leq \epsilon\delta_1\leq \epsilon\delta_1+(1+\epsilon)\sqrt{\frac{\epsilon}{\beta_1}}\delta_{opt}$ (By Lemmas \ref{lem-dis} and \ref{lem-select}).

\noindent\textbf{Induction step:} We assume that there is a path in $\mathcal{T}_i$ from the root to the $j_0$-th level, such that for each $1\leq l\leq j_0$, the level-$l$ node $v_{l}$ on the path is associated with a point $p_{v_{l}}$ satisfying the inequality $||p_{v_l}-m_l ||\leq \epsilon\delta_l+(1+\epsilon)l\sqrt{\frac{\epsilon}{\beta_l}}\delta_{opt} $ (i.e., Induction Hypothesis). Now we consider  the case of  $j=j_0+1$.  Below we will show that there is one child of $v_{j-1}$,  i.e., $v_{j}$, such that  its associated point $p_{v_{l}}$ satisfies the inequality $||p_{v_j}-m_j ||\leq \epsilon\delta_j+(1+\epsilon)j\sqrt{\frac{\epsilon}{\beta_j}}\delta_{opt} $. First, we have the following claim (see Appendix for the proof).

\vspace{-0.1in}
\begin{claim}[\textbf{1}]
In the set of radius candidates built in  Algorithm Sphere-Peeling-Tree, there exists one value $r_j\in \mathcal{R}$ such that
\small{
\vspace{-0.08in}
$$j\sqrt{\frac{\epsilon}{\beta_j}}\delta_{opt}\leq r_j\leq (1+\frac{\epsilon}{2})j\sqrt{\frac{\epsilon}{\beta_j}}\delta_{opt}.$$\vspace{-0.08in}
}\normalsize
\end{claim}
\vspace{-0.1in}

%\begin{figure}[h]
%%\vspace{-0.1in}
%  \centerline{
%  \includegraphics[height=1.3in]{pe1.pdf}}
%  %\vspace{-0.1in}
%    \caption{The $j$ peeling spheres.}
%  \label{fig-pe1}
%  %\vspace{-0.1in}
%  \end{figure}
  
Now, we construct the $j-1$ peeling spheres, $\{B_{j,1}, \cdots, B_{j,j-1}\}$ (as in Algorithm Sphere-Peeling-Tree). For each $1\leq l\leq j-1$, $B_{j,l}$ is centered at $p_{v_l}$ and with radius $r_j$. By Markov inequality and  induction hypothesis, we have the following claim (see Appendix for the proof).
\vspace{-0.05in}
\begin{claim}[\textbf{2}]
For each $1\leq l\leq j-1$, we have $|Opt_l \setminus (\bigcup^{j-1}_{w=1}B_{j,w})|\le \frac{4\beta_j |\mathcal{G}|}{\epsilon}$.

\end{claim}
\vspace{-0.05in}

%due to space limit, we put the computation detail in Appendix, i.e., Section \ref{sec-optl}. 

Claim \textbf{2} shows that $|Opt_l \setminus (\bigcup^{j-1}_{w=1}B_{j,w})|$ is bounded for $1\leq l\leq j-1$, which helps us to find the approximate mean point of $Opt_j $. Induced by the $j-1$ peeling spheres $\{B_{j,1}, \cdots, B_{j,j-1}\}$, $Opt_j$ is divided into $j$ subsets, $Opt_j\cap B_{j,1}$, $\cdots$, $Opt_j\cap B_{j,j-1}$ and $Opt_j \setminus(\bigcup^{j-1}_{w=1}B_{j,w})$. To simplify our discussion, we let $P_l$ denote $Opt_j\cap B_{j,l}$ for $1\leq l\leq j-1$,  $P_j$ denote $Opt_j \setminus(\bigcup^{j-1}_{w=1}B_{j,w})$, and $\tau_{l}$ denote the mean point of $P_l$. Note that the peeling spheres may intersect with each other. For any two intersecting spheres $B_{j,l_1}$ and $B_{j,l_2}$, we let the points set $Opt_j\cap (B_{j,l_1}\cap B_{j,l_2})$ belong to either $P_{l_1}$ or $P_{l_2}$ arbitrarily. Thus, we can assume that $\{P_l\mid 1\leq l\leq j\}$ are pairwise disjoint. Now consider the size of $P_{j}$ (i.e., $|P_j|$). We have the following two cases: (a) $|P_j |\geq \epsilon^3\frac{\beta_{j}}{j}|\mathcal{G}|$ and (b) $|P_j |<\epsilon^3\frac{\beta_{j}}{j}|\mathcal{G}|$. In the following, we show how,  in each case, Algorithm Sphere-Peeling-Tree can obtain an approximate mean point for $Opt_{j}$ by using the Simplex Lemma (i.e., Lemma \ref{lem-shift}).
%\vspace{-0.1in}
%\begin{figure}[ht]
%%\vspace{-0.1in}
%  \centerline{
%  \includegraphics[height=1.3in]{pe1.pdf}}
%  %\vspace{-0.1in}
%    \caption{The $j$ peeling spheres.}
%  \label{fig-pe1}
%  %\vspace{-0.1in}
%  \end{figure}
%
%\begin{figure}[h]
%%\vspace{-0.1in}
%  \centerline{
%  \includegraphics[height=1.3in]{pe3.pdf}}
%  %\vspace{-0.1in}
%    \caption{The $j$ unknown clusters,
%$Opt_{j+1}\setminus\bigcup^{j}_{l=1,l\neq l_{0}}B_{j+1,l}$, $l_{0}=1, \cdots,
%j$. }
%  \label{fig-pe3}
%  %\vspace{-0.1in}
%  \end{figure}
%  
  %%%%%%%%%%%%%%%%%%%%%%%%%%%%%%%%%%%%%%%%%%%%%%%%%%%%%%%%%%%%%%%%hu Dec 1
  
 \vspace{-0.05in}  
  
 \begin{figure}[]
\vspace{-0.25in}
\begin{minipage}[t]{0.5\linewidth}
  \centering
  \includegraphics[height=1.3in]{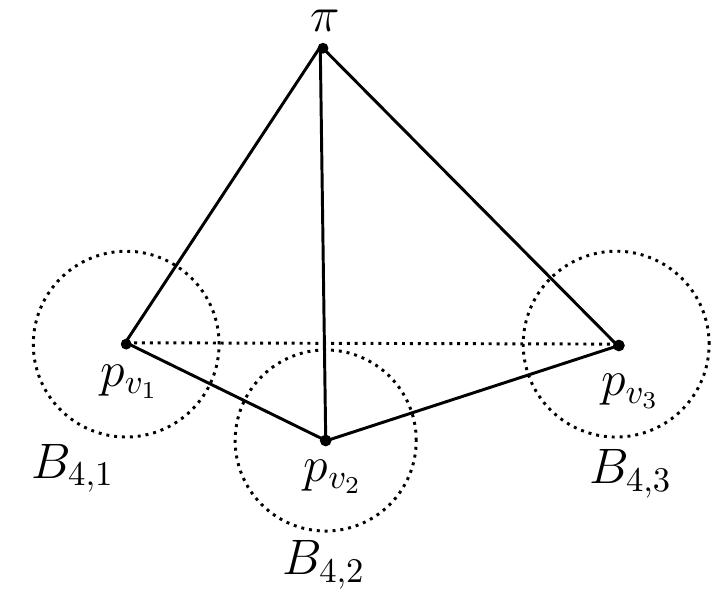}
  \vspace{-0.12in}
     \caption{Case (a) for $j=4$.}
  \label{fig-case1}
\end{minipage}
\begin{minipage}[t]{0.5\linewidth}
\centering
  \includegraphics[height=1in]{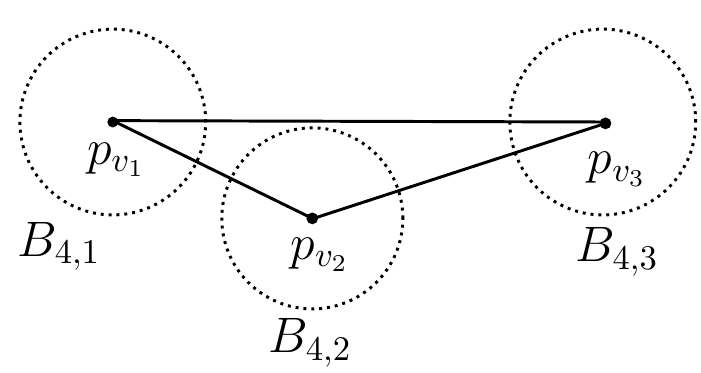}
  \vspace{-0.12in}
   \caption{Case (b) for $j=4$.}
  \label{fig-case2}
  \end{minipage}
  \vspace{-0.2in}
\end{figure}
 \vspace{-0.05in}  
  For case (a),   by Claim \textbf{2}, together with the fact that $\beta_l\leq \beta_{j}$ for $l>j$, we know that 
\small{
\vspace{-0.07in} 
$$\frac{|P_j |}{\sum_{1\leq i\leq k}|Opt_{i}\setminus(\bigcup^{j-1}_{l=1}B_{j,l})|}\geq\frac{\frac{\epsilon^3}{j}\beta_j}{\frac{4(j-1)\beta_j}{\epsilon}+\frac{\epsilon^2}{j}\beta_j+(k-j)\beta_j}>\frac{\epsilon^4}{8kj}\ge\frac{\epsilon^4}{8k^2}.$$
\vspace{-0.03in}
}\normalsize 
This means that $P_{j}$ is large enough, comparing to the set of points outside the peeling spheres. Hence, we can use random sampling technique to obtain an approximate mean point $\pi $ for $P_j$ in the following way. First, we set $t=\frac{k}{\epsilon^5}$, $\eta=\frac{\epsilon}{k}$, and take a sample of size $\frac{t\ln(t/\eta)}{\epsilon^4 /8k^2}=\frac{8k^3}{\epsilon^9}\ln\frac{k^2}{\epsilon^6}$.  By Lemma \ref{lem-select}, we know that  with probability $1-\frac{\epsilon}{k}$, the sample contains $\frac{k}{\epsilon^5}$ points from $P_j$.  Then we let $\pi$ be the mean point of the $\frac{k}{\epsilon^5}$ points from $P_j$, and $a^{2}$ be the variance of $P_j$. By Lemma \ref{lem-dis}, we know that  with probability $1-\frac{\epsilon}{k}$, $||\pi-\tau_j||^2\leq \epsilon^4 a^2$. Also, since $\frac{|P_j|}{|Opt_j|}\geq\frac{\epsilon^3}{j}$, we have $a^2\le\frac{|Opt_j|}{|P_j|}\delta^2_j\leq\frac{j}{\epsilon^3}\delta^2_j$. Thus, $||\pi-\tau_j||^2\leq \epsilon j\delta^2_j$.

Once obtaining $\pi$,  we can now use Lemma \ref{lem-shift} to find a point $p_{v_{j}}$ satisfying the condition of $||p_{v_{j}}-m_j||\leq \epsilon\delta_j+(1+\epsilon)j\sqrt{\frac{\epsilon}{\beta_j}}\delta_{opt}$. First, we construct a simplex $V'_{(a)}$ determined by $\{p_{v_1}, \cdots, p_{v_{j-1}}\}$ and $\pi$ (see Figure. \ref{fig-case1}).  Note that $Opt_j$ is divided by  the peeling spheres into $j$ disjoint subsets, $P_1, \cdots, P_j$, which is a partition of $Opt_{j}$. Each $P_l$ ($1\le l\le j-1$) locates inside $B_{j,l}$, which implies that $\tau_l$ is also inside $B_{j,l}$. Further,  since $||p_{v_l}-\tau_l||\leq r_j\leq (1+\frac{\epsilon}{2})j\sqrt{\frac{\epsilon}{\beta_j}}\delta_{opt}$ for $1\leq l\leq j-1$ (by Claim \textbf{1}), and $||\pi-\tau_j||\leq \sqrt{\epsilon j}\delta_j\leq \sqrt{\frac{\epsilon j}{\beta_j}}\delta_{opt}$ (by $\beta_j\delta^2_j\le\delta^2_{opt}$, which implies $\delta_j\le\sqrt{1/\beta_j}\delta_{opt}$), after setting the value of $L$ (in Lemma \ref{lem-shift}) to be $\max\{r_j,||\pi-\tau_j|| \}\le\max\{(1+\frac{\epsilon}{2})j\sqrt{\frac{\epsilon}{\beta_j}}\delta_{opt}, \sqrt{\frac{\epsilon j}{\beta_j}}\delta_{opt}\}\le(1+\frac{\epsilon}{2})j\sqrt{\frac{\epsilon}{\beta_j}}\delta_{opt}$ and the value of $\epsilon$ (in Lemma \ref{lem-shift}) to be $\epsilon_0=\epsilon^2/4$, by Lemma \ref{lem-shift} we can construct a grid inside the simplex $V'_{(a)}$ with size $O((\frac{8j}{\epsilon_0})^j)$ which ensures the existence of one grid point $\tau$ satisfying the inequality of $||\tau-m_j||\leq \sqrt{\epsilon_0}\delta_j+(1+\epsilon_0)L\leq\epsilon\delta_j+(1+\epsilon)j\sqrt{\frac{\epsilon}{\beta_j}}\delta_{opt}$. Hence, we can use $\tau$ as $p_{v_j}$, and the induction step  holds for this case.

%For each $OPT_{j+1}\bigcap B_{j+1,l}$, we can just use $c_{l}$ as its clustering
%center (i.e., mean point).

%%%%%%%%%%%%%%%%%%%%%%%%%%%%%%%%%%%%%%%%%%%%%%%%%%%%%%%%%%%%%%%%%%%%%%%%%%%%%%%%%%%hu Dec 1
%\begin{figure}[h]
%%\vspace{-0.1in}
%  \centerline{
%  \includegraphics[height=1.3in]{pe3.pdf}}
%  %\vspace{-0.1in}
%    \caption{The $j-1$ unknown clusters,
%$Opt_{j}\setminus\bigcup^{j-1}_{l=1,l\neq l_{0}}B_{j,l}$, $l_{0}=1, \cdots,
%j-1$. }
%  \label{fig-pe3}
%  %\vspace{-0.1in}
%  \end{figure}

For case (b), since $P_{j}$ has a small size, we cannot directly perform  random sampling on it to find its approximate mean point. To overcome this difficulty,  we merge $P_{j}$ with some other large subset $P_{l}$. 
Particularly, since $\sum^{j-1}_{l=1}|P_l |= |Opt_j |-|P_j|\geq (\beta_{j}-\epsilon^3\frac{\beta_{j}}{j})|\mathcal{G}|$, by pigeonhole principle, we know that there exists one $l_0$ such that $P_{l_0}$ has size at least $\frac{1}{j-1}(\beta_{j}-\epsilon^3\frac{\beta_{j}}{j})|\mathcal{G}|$. Without loss of generality, we
assume $l_0=1$. Then $|P_1 |\geq
\frac{1}{j-1}(\beta_{j}-\epsilon^3\frac{\beta_{j}}{j})|\mathcal{G}|$, and we can view $P_1\cup P_j$ as one large enough subset of $Opt_j$. Let $\tau'$ denote the mean point of $P_1\cup P_j$, then we have the following claim (see Appendix for the proof).

\vspace{-0.05in}
\begin{claim}[\textbf{3}]
$||\tau_1-\tau'||\le\frac{\sqrt{2}\epsilon}{1-\epsilon^3}\sqrt{\frac{j\epsilon}{\beta_j}}\delta_{opt}$.
\end{claim}
\vspace{-0.05in}

This means that we can also use Lemma \ref{lem-shift}  to find an approximate mean point in a way similar to case (a) (see Figure. \ref{fig-case2}); the difference is that $Opt_j$ is divided into $j-1$ subsets (i.e., $P_1$ and $P_j$ is viewed as one subset $P_1\cup P_j$) and the value of $L$ is set to be $r_j+||\tau_1-\tau'||\le r_j+\frac{\sqrt{2}\epsilon}{1-\epsilon^3}\sqrt{\frac{j\epsilon}{\beta_j}}\delta_{opt}$. We can first construct a simplex $V'_{(b)}$ determined by $\{p_{v_1}, \cdots, p_{v_{j-1}}\}$ (see Figure. \ref{fig-case2}), and then build a grid inside $V'_{(b)}$ with size $O((\frac{8j}{\epsilon_0})^j)$, where  $\epsilon_0=\epsilon^2/4$.  
% with the case \textcircled{a}, we assign $\epsilon_0=\epsilon^2$, and construct a grid inside the simplex with size $O((\frac{4j}{\epsilon_0})^j)$, 
By Lemma \ref{lem-shift}, we know that there exists one grid point $\tau$ satisfying the condition of $||\tau-m_j||\leq \sqrt{\epsilon_0}\delta_j+(1+\epsilon_0)L\leq\epsilon\delta_j+(1+\epsilon)j\sqrt{\frac{\epsilon}{\beta_j}}\delta_{opt}$. Thus  the induction step holds for this case.
%\begin{figure}[ht]
%%\vspace{-0.1in}
%  \centerline{
%  \includegraphics[height=1.3in]{pe3.pdf}}
%  %\vspace{-0.1in}
%    \caption{The $j$ unknown clusters,
%$OPT_{j+1}\setminus\bigcup^{j}_{l=1,l\neq l_{0}}B_{j+1,l}$, $l_{0}=1, \cdots,
%j$. }
%  \label{fig-pe3}
%  %\vspace{-0.1in}
%  \end{figure}
%This means that we can use $c_{1}$ as the approximate mean point $c_{j+1}$ for $m_{j+1}$.  For all other points in %$OPT_{j+1}$, we
%just let $c_{l}$ be the approximate mean point of $OPT_{j+1} \cap B_{j+1,l}$ for $2\leq l \leq j$.  %(Note that in this case, we only generate $j$ clusters, instead of $j+1$.)

Since Algorithm Sphere-Peeling-Tree executes every step in our above discussion, the induction step, as well as the lemma, is true.
%In summary,the induction step holds. 
\qed
\end{proof}
\vspace{-0.06in}
\noindent\textbf{Success probability:} From the above analysis, we know that in the $j$-th step/iteration, only  case (a) (i.e., $|P_j |\geq \epsilon^3\frac{\beta_{j}}{j}|\mathcal{G}|$) needs to consider success probability, since case (b) (i.e., $|P_j |< \epsilon^3\frac{\beta_{j}}{j}|\mathcal{G}|$) does not need to do sampling. Recall that in case (a), we take a sample of size $\frac{8k^3}{\epsilon^9}\ln\frac{k^2}{\epsilon^6}$. Thus  with probability $1-\frac{\epsilon}{k}$, it contains $\frac{k}{\epsilon^5}$ points from $P_j$. Meanwhile, with probability $1-\frac{\epsilon}{k}$, $||\pi-\tau_j||^2\leq \epsilon^4 a^2$. Hence, the success probability in the $j$-th step is $(1-\frac{\epsilon}{k})^2$, which means that the success probability in all $k$ steps is $(1-\frac{\epsilon}{k})^{2k}\geq 1-2\epsilon$.

\noindent\textbf{Running time:} Algorithm $k$-CMeans calls Algorithms Sphere-Peeling-Tree $\frac{2k}{\epsilon}$ times. It is easy to see that each node on the tree returned from Algorithm Sphere-Peeling-Tree has $|\mathcal{R}|2^m (\frac{32j}{\epsilon^2})^j$ children, where $|\mathcal{R}|=O(\frac{\log kn}{\epsilon})$, and $m=\frac{8k^3}{\epsilon^9}\ln\frac{k^2}{\epsilon^6}$. Since the tree has a height of $k$, the complexity of the tree is $O(2^{poly(\frac{k}{\epsilon})}(\log n)^k)$. Further, since each node takes $O(|\mathcal{R}|2^m (\frac{32j}{\epsilon^2})^j nd)$ time, the total time complexity of Algorithm $k$-CMeans is $O(2^{poly(\frac{k}{\epsilon})}n(\log n)^{k+1} d )$.
\vspace{-0.16in}
\section{Extension to Chromatic $k$-Medians Clustering}
\label{sec-extension}
\vspace{-0.08in}
We extend our ideas for $k$-CMeans to the Chromatic  $k$-Medians Clustering problem ($k$-CMedians). Similar to $k$-CMeans, we first show its relationship with $k$-medians, and then present a $(5+\epsilon)$-approximation algorithm using the sphere peeling technique. Due to the lack of a similar Simplex Lemma for $k$-CMedians, we achieve a constant approximation, instead of a PTAS. See details of the algorithm in  Section \ref{sec-kmedian} of the Appendix.

% Due to space limit and similarity of techniques, we put the details in Section \ref{sec-kmedian} of Appendix.

%Finally, combining Lemma \ref{lem-induction}, Lemma \ref{lem-equal} and the above running time analysis, we get theorem \ref{the-ptas}.

%\section{Appendix}

\section{Figure. \ref{fig-probe}}

\begin{figure}[ht]
\vspace{-0.2in}
  \centerline{
  \includegraphics[scale=0.3]{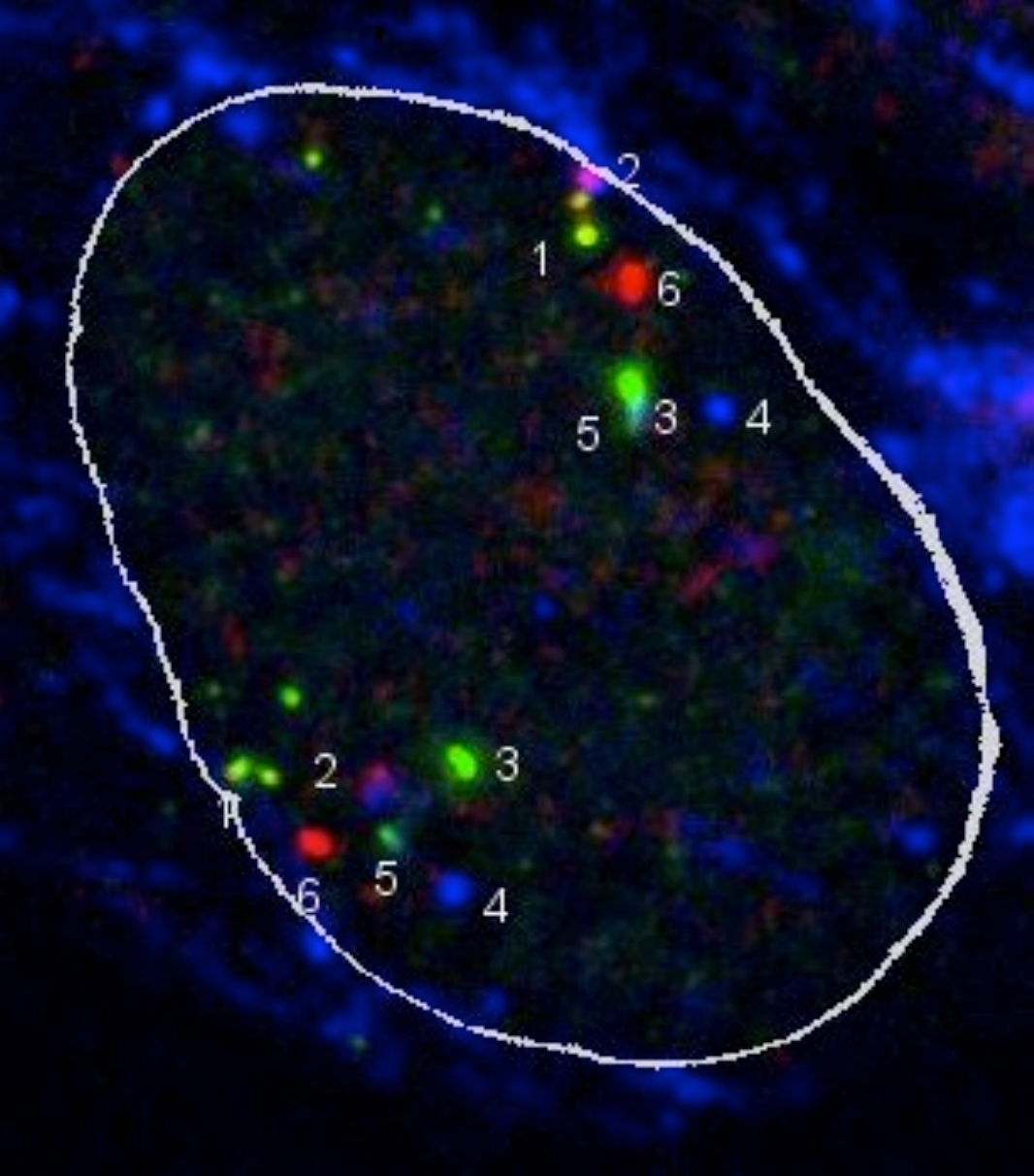}}
  \vspace{-0.15in}
    \caption{BAC probes of Chromosome 1 in a WI38 cell with homolog having 6 probes.}
  \label{fig-probe}
  \vspace{-0.2in}
\end{figure}

\section{Proof for Theorem \ref{mfptas}}
\label{sec-mfptas}
 \begin{proof}
Since it is sufficient to show that the theorem holds for the case of $k=2$, we assume in this proof that $k=2$ and each point-set $G_{i}$ has exactly two points. We make use of a construction by Dasgupta for the
NP-hardness proof of the 2-mean clustering problem in high dimensional space \cite{D08}. Their proof reduces from the NAE3SAT problem. For better understanding our
ideas, below we sketch their construction.
%In \cite{D08}, Dasgupta proved that in high dimension space $2-$ means clustering is NP-hard by reducing NAE3SAT to it. The sketch of the proof in \cite{D08} is as follows:
\begin{enumerate}
\item For any instance $\phi$ of NAE3SAT with  literal set $\{x_{1}, \cdots, x_{n}\}$ and $m$ clauses,  construct a $2n \times 2n$ matrix $D_{\alpha,\beta}$ as follows, where the indices  correspond to  $\{x_{1}, \cdots, x_{n}\}$ when they are in the range of $[1,n]$, and to $\{\overline{x_{1}}, \cdots, \overline{x_{n}}\}$ when they are in the range of $[n+1,2n]$.\\
$D_{\alpha, \beta} = \left\{ \begin{array}{ll}
0 & \textrm{if $\alpha=\beta$}\\
1+\Delta & \textrm{if $\alpha=\overline{\beta}$}\\
1+\delta & \textrm{if $\alpha\sim\beta$}\\
1 & \textrm{otherwise,}
\end{array} \right.$
%where $\alpha, \beta$
%For any index $\alpha \in \{1, \cdots, n\}$, it corresponds to $\{x_{1}, \cdots, x_{n}\}$, and $\alpha\in\{n+1, \cdots, 2n\}$ corresponds to $\{\overline{x_{1}}, \cdots, \overline{x_{n}}\}$.
%\item
%$D_{\alpha, \beta} = \left\{ \begin{array}{ll}
%0 & \textrm{if $\alpha=\beta$}\\
%1+\Delta & \textrm{if $\alpha=\overline{\beta}$}\\
%1+\delta & \textrm{if $\alpha\sim\beta$}\\
%1 & \textrm{otherwise}
%\end{array} \right.$
where $\Delta$, $\delta$ are two constants satisfying inequalities $0<\delta<\Delta<1$ and $4\delta m<\Delta\leq 1-2\delta n$, and $\alpha\sim\beta$ means that both $\alpha$ and $\beta$ or both $\overline{\alpha}$ and $\overline{\beta}$ appear in a clause.

\item $D$ can be embedded into $R^{2n}$, i.e., there exist $2n$ points in $R^{2n}$ with $D$ as their distance matrix.

\item Let $C_{1}$ and $C_{2}$ be the two clusters of the $2$-mean clustering of the $2n$ embedding points. If for any $i$, the points corresponding to
%If classify the $2n$ index into $2$ classes $C_{1}, C_{2}$, and for any $i$, the index corresponding to
$x_{i}$ and $\overline{x_{i}}$ are separated into different clusters, then $\phi$ is satisfiable if and only if $$\frac{1}{2n}\sum_{i,j\in C_{1}}D_{i,j}+\frac{1}{2n}\sum_{i,j\in C_{2}}D_{i,j}\leq n-1+\frac{2\delta m}{n}.$$

\item Since $\frac{1}{2n}\sum_{i,j\in C_{1}}D_{i,j}+\frac{1}{2n}\sum_{i,j\in C_{2}}D_{i,j}$ is the total cost of the $2$-mean clustering for $C_{1}$ and $C_{2}$,  a polynomial time solution to the $2$-mean clustering problem in high dimensional space implies a polynomial time solution to NAE3SAT. Thus the 2-mean clustering is NP-hard in high dimensions.
\end{enumerate}

The above reduction can be naturally extended to show the NP-hardness of the full chromatic $2$-mean clustering problem.   To show this, we only need to construct $G_{i}$ as the set containing the two points  corresponding to $x_{i}$ and $\overline{x_{i}}$ (for simplicity, we write it as $G_{i}=\{x_{i}, \overline{x_{i}}\}$), and the remaining proof follows from the same argument.

Next, we show that full $2$-CMean has no FPTAS in high dimensional space unless P=NP. To see this, we still use the same construction. From the above discussion,  we know that $\phi$ is unsatisfiable if and only if for any chromatic partition of $\mathcal{G}$, there exists one clause in $\phi$ such that the three points corresponding to the three literals in this clause are clustered into the same cluster. Hence, the total cost for any chromatic partition is at least
$$2\frac{1}{n}({n\choose 2}+(m-1)\delta+3\delta)=n-1+\frac{2}{n}(m+2)\delta .$$

The ratio $\eta$ between the minimum chromatic partition cost of an unsatisfiable instance and the upper bound cost of a satisfiable instance is  $$\eta=\frac{n-1+\frac{2}{n}(m+2)\delta}{n-1+\frac{2\delta m}{n}}=1+\frac{\frac{4}{n}\delta}{n-1+\frac{2(m+2)}{n}\delta}.$$ If we let $\delta=\frac{1}{5m+2n}$, then $\eta=1+\frac{\frac{4}{n}\delta}{n-1+\frac{2(m+2)}{n}\delta}=1+\frac{4}{n(5m+2n)(n-1)+2(m+2)}$.

Suppose that there exists an FPTAS for the full chromatic $2$-means clustering problem. Then, if we let $\epsilon <\frac{4}{n(5m+2n)(n-1)+2(m+2)}$,  the cost of a $(1+\epsilon)$-approximation of the full $2$-CMeans is  less than $n-1+\frac{2}{n}(m+2)\delta$ if and only if $\phi$ is satisfiable. Since the running time of the FPTAS for full $2$-CMeans and $\frac{1}{\epsilon}$ are all  polynomial functions of $m$ and $n$,  this implies that  NAE3SAT can be solved in polynomial time. Obviously this can only happen if P=NP.
\qed
\end{proof}

\section{Proof of Lemma \ref{lem-meanshift}}
\label{sec-meanshift}
\begin{proof}
In the our following discussion, we use $<a, b>$ to denote the inner product of $a$ and $b$. It is easy to see that
$$\sum_{p\in P}||p-m'||^2=\sum_{p\in P}||p-m+m-m'||^2$$
$$=\sum_{p\in P}(||p-m||^2+2<p-m,m-m'>+||m-m'||^2)$$
$$=\sum_{p\in P}||p-m||^2+2\sum_{p\in P}<p-m,m-m'>+|P| \times ||m-m'||^2$$
$$=\sum_{p\in P}||p-m||^2+2<\sum_{p\in P}(p-m),m-m'>+|P| \times ||m-m'||^2.$$
Since $m$ is the mean point of $P$,  $\sum_{p\in P}(p-m)=0$. Thus, the above equality becomes $\sum_{p\in P}||p-m'||^2=\sum_{p\in P}||p-m||^2+|P| \times ||m-m'||^2$.
\qed
\end{proof}

\section{Proof of Lemma \ref{lem-shift}}
\label{sec-shift}
\begin{proof}
Similar to Lemma \ref{lem-simplex}, we prove this lemma by mathematics induction on $j$. 

\textbf{Base case.} For $j=1$, since $o_1=o$, we just need to let $\tau=o'_1$. Then, we have $||\tau-o||=||o'_1-o||=||o'_1-o_1||\leq L\leq\sqrt{\epsilon}\delta+(1+\epsilon)L$. Thus, the base case holds.

\textbf{Induction step.} Assume that the lemma holds for any $j\leq j_0$ for some $j_{0} \ge 1$ (i.e., Induction Hypothesis). Now we consider the case of $j=j_0+1$. Similar to the proof of Lemma \ref{lem-simplex}, we assume that $\frac{|P_l |}{|P|}\geq \frac{\epsilon}{4j}$ for each $1\leq l\leq j$. Otherwise, it can be reduced to a problem with smaller $j$, and solved by the induction hypothesis.  Hence, in the following discussion, we assume that $\frac{|P_l |}{|P|}\geq \frac{\epsilon}{4j}$ for each $1\leq l\leq j$. 

First, we know that $o=\sum^j_{l=1}\frac{|P_l |}{|P|} o_l$. Let $o'=\sum^j_{l=1}\frac{|P_l |}{|P|} o'_l$. Then, we have 
\begin{eqnarray}
||o-o'||=||\sum^j_{l=1}\frac{|P_l |}{|P|} o_l-\sum^j_{l=1}\frac{|P_l |}{|P|} o'_l||\leq \sum^j_{l=1}\frac{|P_l |}{|P|}||o_l-o'_l||\leq L. \label{for-8}
\end{eqnarray}

Thus, if we can find a grid point $\tau$ within a distance to $o'$ no more than $\sqrt{\epsilon}\delta+\epsilon L$ (i.e., $||\tau-o'||\leq \sqrt{\epsilon}\delta+\epsilon L$), by inequality (\ref{for-8}), we will have $||\tau-o||\leq||\tau-o'||+||o'-o||\leq\sqrt{\epsilon}\delta+(1+\epsilon)L$. This means that  we only need to find a grid point close enough to $o'$. 

To find such a $\tau$, we first consider the distance from $o'_{l}$ to $o'$. For any $1\leq l\leq j$, we have
\begin{eqnarray}
||o'_l-o'||\leq ||o'_l-o_l||+||o_l-o||+||o-o'||\leq 2\sqrt{\frac{j}{\epsilon}}\delta+2L, \label{for-9}
\end{eqnarray}
where the first inequality follows from triangle inequality, and the second inequality follows from the facts that $||o'_l-o_l||$ and $||o-o'||$ are both bounded by $L$, and $||o_l-o||\leq  2\sqrt{\frac{j}{\epsilon}}\delta$ (by Lemma \ref{lem-close}).

 This implies that we can use the similar idea in Lemma \ref{lem-simplex} to construct a ball $\mathcal{B}$ centered at any  $o'_{l_0}$ and with radius $r=\max_{1\leq l\leq j}\{||o'_l-o'_{l_0}||\}$. Note that  since $||o'_l-o'_{l_0}||\leq ||o'_l-o'||+||o'-o'_{l_0}||\leq 4\sqrt{\frac{j}{\epsilon}}\delta+4L$ (by inequality (\ref{for-9})), the simplex $V'$ is inside $\mathcal{B}$. Similar to Lemma \ref{lem-simplex}, we can build a grid inside $\mathcal{B}$ with grid length $\frac{\epsilon r}{4j}$ and total grid points $O((8j/\epsilon)^j)$. Clearly in this grid, we can find a grid point $\tau$ such that $||\tau-o'||\leq \frac{\epsilon}{4\sqrt{j}}r\leq \sqrt{\epsilon}\delta+\epsilon L$. Thus, $||\tau-o||\leq\sqrt{\epsilon}\delta+(1+\epsilon)L$, and the induction step, as well as the lemma, holds.
\qed
\end{proof}

\section{Proof of Claim 1 in Lemma \ref{lem-induction}}
\label{sec-claim1}
%In Algorithm Sphere-Peeling-Tree, if the input $\delta\in [\delta_{opt}, (1+\epsilon)\delta_{opt}]$, then for any fixed $j$, the set of $\{2^t\sqrt{\epsilon}\delta\mid t=0, 1, \cdots, \frac{\log(kn)}{2}\}$ must contain one value $\hat{r}_j$ such that
%\small{
%\vspace{-0.07in}
%$$\sqrt{\frac{\epsilon}{\beta_j}}\delta_{opt}\leq \hat{r}_j\leq 2(1+\epsilon)\sqrt{\frac{\epsilon}{\beta_j}}\delta_{opt}.$$
%\vspace{-0.1in}
%}\normalsize

\begin{proof}
Since $1\geq \beta_j\geq \frac{1}{|\mathcal{G}|}\geq\frac{1}{kn}$, there is one integer $t$ between $1$ and $\log(kn)$, such that $2^{t-1}\leq\frac{1}{\beta_j}\leq 2^t$. Thus $ 2^{t/2-1}\sqrt{\epsilon}\delta_{opt}\leq\sqrt{\frac{\epsilon}{\beta_j}}\delta_{opt}\leq 2^{t/2}\sqrt{\epsilon}\delta_{opt}$. Together with $\delta\in [\delta_{opt}, (1+\epsilon)\delta_{opt}]$, we have
%\small{
%\vspace{-0.05in}
$$ 2^{t/2-1}\sqrt{\epsilon}\frac{\delta}{1+\epsilon}\leq\sqrt{\frac{\epsilon}{\beta_j}}\delta_{opt}\leq 2^{t/2}\sqrt{\epsilon}\delta .$$
%\vspace{-0.03in} 
 %}\normalsize
 Thus if  set $\hat{r}_j=2^{t/2}\sqrt{\epsilon}\delta$, we have $\sqrt{\frac{\epsilon}{\beta_j}}\delta_{opt}\leq \hat{r}_j\leq 2(1+\epsilon)\sqrt{\frac{\epsilon}{\beta_j}}\delta_{opt}$. Let $x=\frac{j\hat{r}_j}{j\sqrt{\frac{\epsilon}{\beta_j}}\delta_{opt}}$. Then we have $1\le x\le 2(1+\epsilon)$.  We build a grid in the interval $[\frac{x}{2(1+\epsilon)}, x]$ with the grid length $\frac{\epsilon}{4(1+\epsilon)}x$, and obtain a grid set (i.e., number set) $\mathcal{N}=\{\frac{1+l\frac{\epsilon}{2}}{2(1+\epsilon)}x\mid 0\le l\le 4+\frac{2}{\epsilon}\}$. We prove that there must exist one number in $\mathcal{N}$ and is between $1$ and $1+\epsilon/2$. First, we know that $\frac{x}{2(1+\epsilon)}\le 1\le x$. If $x\le 1+\epsilon/2$, we find the the desired number $x$ in $\mathcal{N}$. Otherwise, the whole interval $[1, 1+\epsilon/2]$ is inside $[\frac{x}{2(1+\epsilon)}, x]$. Since the grid has grid length $\frac{\epsilon}{4(1+\epsilon)}x\le\frac{\epsilon}{4(1+\epsilon)}2(1+\epsilon)=\epsilon/2$, there must exist one grid point locating inside $[1, 1+\epsilon/2]$. Thus, the desired number exists in $\mathcal{N}$.

  Let $\mathcal{R}_j=\{\frac{1+l\frac{\epsilon}{2}}{2(1+\epsilon)}j\hat{r}_j\mid 0\le l\le 4+\frac{2}{\epsilon}\}$. From the above analysis, we know that there exists one value $r_j\in \mathcal{R}_j$ such that
%\small{
%\vspace{-0.03in}
$$j\sqrt{\frac{\epsilon}{\beta_j}}\delta_{opt}\leq r_j\leq (1+\frac{\epsilon}{2})j\sqrt{\frac{\epsilon}{\beta_j}}\delta_{opt}.$$\vspace{-0.08in}
%}\normalsize

Note that $\mathcal{R}_j\subset \mathcal{R}$, where $\mathcal{R}=\cup^{\log(kn)}_{t=0}\{\frac{1+l\frac{\epsilon}{2}}{2(1+\epsilon)}j2^{t/2}\sqrt{\epsilon}\delta\mid 0\le l\le 4+\frac{2}{\epsilon}\}$. Thus, the Claim is proved.

%is the set of radius candidates built in  Algorithm Sphere-Peeling-Tree. Thus there exists one value $r_j\in \mathcal{R}_j$ such that
%\small{
%%\vspace{-0.03in}
%$$j\sqrt{\frac{\epsilon}{\beta_j}}\delta_{opt}\leq r_j\leq (1+\frac{\epsilon}{2})j\sqrt{\frac{\epsilon}{\beta_j}}\delta_{opt}.$$\vspace{-0.08in}
%}\normalsize
%
%Thus
\qed
\end{proof}

\section{Proof of Claim 2 in Lemma \ref{lem-induction}}
\label{sec-optl}
\begin{proof}
First, for each $1\leq l\leq j-1$, we have $|Opt_l \setminus (\bigcup^{j-1}_{w=1}B_{j,w})|\leq|Opt_l \setminus B_{j,l} | $. Secondly, by Markov inequality, we have
%\small{
$$|Opt_l \setminus B_{j,l} | \leq \frac{\delta^2_l}{(r_j-||p_{v_l}-m_l ||)^2}|Opt_l |.$$ 
%}\normalsize

%\begin{claim}[\textbf{2}]
%For any $1\leq j\leq k$, $\delta^2_j\leq \frac{1}{\beta_j}\delta^2_{opt}$.
%\end{claim}
%\begin{proof}
%It directly follows from the fact that $\delta^2_{opt}=\sum^k_{j=1}\beta_j\delta^2_j$.
%\qed
%\end{proof}

%\noindent{\bf Upper Bounding $|Opt_l \setminus B_{j,l} | $:\\}

Note that $\delta^2_{opt}=\sum^k_{j=1}\beta_j\delta^2_j$, and $\beta_j\le \beta_l$ (by $l<j$). Thus, we have $\delta_{l}\leq \sqrt{\frac{1}{\beta_l}}\delta_{opt}\le \sqrt{\frac{1}{\beta_j}}\delta_{opt}$. Together with $j\sqrt{\frac{\epsilon}{\beta_j}}\delta_{opt}\leq r_j$ and $||p_{v_l}-m_l ||\leq \epsilon\delta_l+(1+\epsilon)l\sqrt{\frac{\epsilon}{\beta_l}}\delta_{opt} $ (by induction hypothesis), we have 

\begin{eqnarray*}
r_j-||p_{v_l}-m_l ||&\ge& j\sqrt{\frac{\epsilon}{\beta_j}}\delta_{opt}-(\epsilon\delta_l+(1+\epsilon)(j-1)\sqrt{\frac{\epsilon}{\beta_l}}\delta_{opt} )\\
& =& (1-(j-1)\epsilon)\sqrt{\frac{\epsilon}{\beta_j}}\delta_{opt}-\epsilon\delta_l\\
&\geq& (1-(j-1)\epsilon-\sqrt{\epsilon})\sqrt{\frac{\epsilon}{\beta_j}}\delta_{opt}, 
\end{eqnarray*}

where the last inequality follows from $\delta_{l}\leq \sqrt{\frac{1}{\beta_l}}\delta_{opt}\le \sqrt{\frac{1}{\beta_j}}\delta_{opt}$. Thus, we have

\begin{eqnarray*}
|Opt_l \setminus B_{j,l} | &\leq& \frac{\delta^2_l}{(1-(j-1)\epsilon-\sqrt{\epsilon})^2\frac{\epsilon}{\beta_j}\delta^2_{opt}}|Opt_l |\\
&\leq& \frac{\delta^2_l}{(1-(j-1)\epsilon-\sqrt{\epsilon})^2\frac{\epsilon}{\beta_j}\beta_l \delta^2_l}|Opt_l | \\
&=& \frac{\beta_j}{(1-(j-1)\epsilon-\sqrt{\epsilon})^2\epsilon\beta_l}|Opt_l |\\
&=&\frac{\beta_j |\mathcal{G}|}{(1-(j-1)\epsilon-\sqrt{\epsilon})^2\epsilon}\leq\frac{\beta_j |\mathcal{G}|}{(1-j\sqrt{\epsilon})^2\epsilon},
\end{eqnarray*}
%For the induction step (i.e., $j+1$-th iteration), if
%$\delta^2_{j+1}\leq O(\frac{\epsilon}{k\beta_{j+1}})\delta^2_{opt}$,
%we just need to let $p^{1}_{O(j+1)}$ (which can be obtained by Theorem \ref{the-constant}) be the approximate mean point
%for $OPT_{j+1}$. This  increases the cost  by at most
%$(\sqrt{k/\lambda}\delta_{j+1})^2\beta_{j+1}N\leq
%O(\epsilon)N\delta^2_{opt}$  (comparing to $\beta_{j+1}N\delta^2_{opt}$), and $||c_{j+1}-M_{j+1}||\leq\sqrt{\frac{k}{\lambda}}\delta_{j+1}\leq O(\sqrt{k})$ $\delta_{j+1}$.  Hence, the
%induction holds for this case.
%%in the $j+1$ step.
%
%Thus we can assume that $\delta^2_{j+1}>O(\frac{\epsilon}{k\beta_{j+1}})\delta^2_{opt}$. We also know that
%$\delta^2_{j+1}\leq\frac{1}{\beta_{j+1}}\delta^2_{opt}$.  By a
%similar idea in Lemma \ref{ca2}, using the $(1+\frac{k}{\lambda})$-approximation from Theorem \ref{the-constant}, we can generate
%%$O(k\log n)$
%$\frac{\log(kn)}{2}\sqrt{\frac{k}{\epsilon}}$ candidates,
%$\bigcup^{\sqrt{k/\epsilon}}_{l=1}\{l\epsilon\delta, 2l\epsilon\delta, \cdots, 2^{\frac{\log(kn)}{2}}l\epsilon\delta\}$,
%to find the proper value of the radius $r=O(\epsilon)\delta_{j+1}$ of the $j$ peeling spheres. After obtaining the value of $r$, we
%%Assume we get the right value for $r$, and
% construct $j$ peeling spheres $\{B_{j+1,1}, \cdots, B_{j+1,j}\}$, all of radius $r$ and with $B_{j+1,l}$ centered
% at $c_{l}$ for $1\leq l\leq j$ (see Figure \ref{fig-pe1}).  Based on the relation of $OPT_{j+1}$ and the $j$ peeling spheres, we have
%two cases to consider.
%\vspace{-0.1in}
where the second inequality follows from the fact that $\beta_l \delta^2_l\leq \delta^2_{opt}$, and the fourth equation follows from that $\frac{|Opt_l|}{\beta_l}=|\mathcal{G}|$. Note that we can assume $\epsilon$ is small enough such that $\epsilon\leq \frac{1}{4j^2}$, which implies that $ \frac{\beta_j |\mathcal{G}|}{(1-j\sqrt{\epsilon})^2\epsilon}\leq \frac{4\beta_j |\mathcal{G}|}{\epsilon}$. Otherwise, we can just replace $\epsilon$ by $\frac{\epsilon}{4j^2}$ as the input at the beginning of the algorithm. Thus, in total, we have
\begin{eqnarray*}
|Opt_l \setminus B_{j,l} | &\leq& \frac{4\beta_j |\mathcal{G}|}{\epsilon}.
\end{eqnarray*}

Thus the Claim is proved.
\qed
\end{proof}

\section{Proof of Claim 3 in Lemma \ref{lem-induction}}
\label{sec-claim3}

\begin{proof}

First, we have
%\small{
%\vspace{-0.07in} 
 $$\frac{|P_1|}{|P_1\cup P_j |}\geq \frac{\frac{1}{j-1}(1-\frac{\epsilon^3}{j})}{\frac{1}{j-1}(1-\frac{\epsilon^3}{j})+\frac{\epsilon^3}{j}}>\frac{1-\epsilon^3}{1+\epsilon^3}.$$
%\vspace{-0.05in}
%}\normalsize 
 Let $a^{2}$ denote the variance of $P_1\cup P_j$. By Lemma \ref{lem-close}, we know that $||\tau_1-\tau'||\leq\sqrt{\frac{2\epsilon^3}{1-\epsilon^3}}a$. Meanwhile, since $\frac{|P_1\cup P_j |}{|Opt_j|}\geq\frac{ |P_1|}{|Opt_j|}\geq\frac{ \frac{1}{j-1}(\beta_{j}-\epsilon^3\frac{\beta_{j}}{j})|\mathcal{G}|}{\beta_j |\mathcal{G}|}=\frac{1-\frac{\epsilon^3}{j}}{j-1}$, we have $a^2\le \frac{|Opt_j|}{|P_1\cup P_j|}\delta^2_j\leq \frac{j-1}{1-\frac{\epsilon^3}{j}}\delta^2_j$. Then we have
 
 \begin{eqnarray*}
||\tau_1-\tau'||&\leq&\sqrt{\frac{2\epsilon^3}{1-\epsilon^3}}a\le\sqrt{\frac{2\epsilon^3}{1-\epsilon^3}}\sqrt{\frac{j-1}{1-\frac{\epsilon^3}{j}}}\delta_j\\
&\leq&\sqrt{\frac{2j\epsilon^3}{(1-\epsilon^3)(1-\frac{\epsilon^3}{j})\beta_j}}\delta_{opt}\\
&\le&\sqrt{\frac{2j\epsilon^3}{(1-\epsilon^3)(1-\epsilon^3)\beta_j}}\delta_{opt}=\frac{\sqrt{2}\epsilon}{1-\epsilon^3}\sqrt{\frac{j\epsilon}{\beta_j}}\delta_{opt}, 
\end{eqnarray*}

where the third inequality follows from $\delta_j\le \sqrt{\frac{1}{\beta_j}}\delta_{opt}$. Thus, the claim is true.
\qed
\end{proof}
\section{Chromatic $k$-Medians Clustering}
\label{sec-kmedian}

In this section, we extend our ideas for $k$-CMeans to the Chromatic  $k$-Medians Clustering problem ($k$-CMedians).  Similar to $k$-CMeans, we first show its relationship with $k$-medians (in Section \ref{sec-conmedian}), and then present a $(5+\epsilon)$-approximation algorithm (in Section \ref{sec-ptasmedian}). Due to the lack of a similar Simplex Lemma for $k$-CMedians, we achieve a constant approximation, instead of a PTAS.
%We show the formal definition for {\em Chromatic Median Clustering} here:

\begin{definition}[Chromatic $k$-Median Clustering ($k$-CMedians)]
\label{def-kcmedian}
Let $\mathcal{G}=\{G_{1}, \cdots,$ $G_{n}\}$ be a set of $n$ point-sets with each $G_{i}=\{p^{i}_{1}, \dots, p^{i}_{k_{i}}\}$ consisting of $k_{i}\leq k$ points in $\mathbb{R}^d$ space. The chromatic  $k$-median clustering (or $k$-CMedians) of $\mathcal{G}$ is to find $k$ points $\{m_{1}, \cdots, m_{k}\}$ in $\mathbb{R}^d$ space and a chromatic partition $U_{1}, \cdots, U_{k}$ of $\mathcal{G}$ such that $\frac{1}{n}\sum_{j}\sum_{q\in U_{j}}||q-m_{j}||$ is minimized.
\end{definition}

%It is easy to see that $k$-medians is a special case of $k$-CMedians (i.e., each $G_i$ contains exactly one point). As shown by Dasgupta \cite{D08} that there is no FPTAS for $k$-means in high dimensional space even if $k=2$ (unless P=NP), we immediately have the following theorem.
%
% % \vspace{-0.1in}
%
%\begin{theorem}
%\label{the-medianfptas}
%$k$-CMedians  is NP-hard and has no FPTAS for $k\ge 2$ in high dimensional space unless P=NP.
%\end{theorem}

\subsection{Constant Approximation from $k$-Medians}
\label{sec-conmedian}

%\textbf{Notation.} 
Given a set of points in $\mathbb{R}^d$, the optimal median point is also called {\em Fermat Weber point} in geometry. Its main difference with mean point is that no explicit formula exists for computing the optimal median point, while the mean point is simply the average of the given points. Consequently,  median point is often approximated using some iterative procedure, such as {\em Weiszfeld's algorithm}. 
Thus in the following discussion, we only assume the availability of a $(1+\epsilon)$-approximation of the median point. 
%(over the cost, i.e., the sum of distances to the given points). \\

%Similar with Lemma \ref{lem-close}, we have the following lemma:

\begin{lemma}
\label{lem-closemedian}
Let $P$ be a set of points in $\mathbb{R}^d$ space, and $P_{1}$ be a subset of $P$ containing a fraction of $\alpha\leq 1$ points of $P$. Let $m_{opt}$ and $m$ be the optimal and $(1+\epsilon)$-approximate median point of $P$ respectively, and $m_1$ be the optimal median of $P_{1}$. Then  $||m_{1}-m||\leq\frac{2+\epsilon}{\alpha}\mu$, where $\mu=\frac{1}{|P|}\sum_{p\in P}||p-m_{opt}||$.
\end{lemma}
\begin{proof}
Let $\mu_1=\frac{1}{|P_1|}\sum_{p\in P_1}||p-m_1||$.  Since $P_1\subseteq P$, it is easy to know that $\sum_{p\in P}||p-m||\geq \sum_{p\in P_1}||p-m||$, which implies that $(1+\epsilon)\mu\geq\frac{1}{|P|}\sum_{p\in P_1}||p-m||=\alpha\frac{1}{|P_1|}\sum_{p\in P_1}||p-m||$.  By triangle inequality, we also have $||p-m||\geq ||m-m_1||-||p-m_1||$. Thus, 
\begin{eqnarray}
(1+\epsilon)\mu\geq \alpha(||m-m_1||-\mu_1). \label{for-12}
\end{eqnarray}
Since $m_1$ is the optimal median of $P_1$, we have $\mu_1=\frac{1}{|P_1|}\sum_{p\in P_1}||p-m_1||\leq \frac{1}{|P_1|}\sum_{p\in P_1}||p-m_{opt}||\leq\frac{1}{|P_1|}\sum_{p\in P}||p-m_{opt}||=\frac{1}{\alpha}\mu$. Plugging this into inequality (\ref{for-12}), we have $||m-m_1||\leq\frac{2+\epsilon}{\alpha}\mu$.
\qed
\end{proof}

%Similar with idea for $k$-CMeans, we have the following result for $k$-CMedians:
\begin{theorem}
\label{the-constantmedian}
Let $\mathcal{G}=\{G_1, \cdots, G_n\}$ be an instance of $k$-CMedians, and $\mathcal{C}$ be the $k$ $(1+\epsilon)$-approximate median points of the $k$ clusters generated by a $c$-approximation $k$-medians algorithm on the points $\cup_{i=1}^{n}G_{i}$. Then,  $[\mathcal{C}]^k$ contains at least one $k$-tuple whose elements are the $k$ median points of a $((2+\epsilon)ck^2+(2+\epsilon)k+1)$-approximation of $k$-CMedians on 
%a the $k$ $(1+\epsilon)$-approximate median points of the $k$ clusters
%Assume there is an $c$-approximation solution for $k$-medians clustering on $\cup^n_{i=1}G_i$ without the chromatic requirement, and $\mathcal{C}$ contains the $k$ $(1+\epsilon)$-approximate median points of the $k%$ clusters. Then $[\mathcal{C}]^k$ includes one $k$-tuple, which yields $((2+\epsilon)ck^2+(2+\epsilon)k+1)$-approximation for $k$-CMedians on 
$\mathcal{G}$, where $[\mathcal{C}]^k=\underbrace{\mathcal{C}\times\cdots\times\mathcal{C}}_{k}$.
\end{theorem}
\begin{proof}
Let $\{c_1, \cdots, c_k\}$ be the set of $k$ approximate median points in $\mathcal{C}$, and $\{S_1, \cdots, S_k\}$  be the $k$ clusters returned by  the $c$-approximation $k$-medians algorithm. 
%clustering as $\{S_1, \cdots, S_k\}$. 
Thus, $c_j$ is the $(1+\epsilon)$-approximate median point of $S_j$ for $1\leq j\leq k$.  Let  $\mathcal{OPT}=\{Opt_{1}, \cdots, Opt_{k}\}$ be the \textbf{unknown} optimal solution for $k$-CMedians on $\mathcal{G}$, and $m_{j}$ be the optimal median point of $Opt_j$ for $1\leq j\leq k$. Denote the set $Opt_i\cap S_j$ as $\Gamma^i_j$, and its optimal median point as $\tau^i_j$ for $1\leq i, j\leq k$ .

 Since $\cup^k_{j=1}\Gamma^i_j=Opt_i$, there must exist some index $1\leq j_i \leq k$ such that $|\Gamma^i_{j_i}|\geq \frac{1}{k}|Opt_i|$. Fixing $j_i$, we have  the following about $\sum_{p\in Opt_i}||p-c_{j_i}||$. 
\begin{eqnarray*} 
 \sum_{p\in Opt_i}||p-c_{j_i}|| &=& \sum_{p\in Opt_i}||p-m_i+m_i-c_{j_i}||\\
&\leq& \sum_{p\in Opt_i}(||p-m_i||+||m_i-c_{j_i}||)\\
&=& \sum_{p\in Opt_i}||p-m_i||+|Opt_i| \times ||m_i-c_{j_i}||\\
&=& \sum_{p\in Opt_i}||p-m_i||+|Opt_i| \times ||m_i-\tau^i_{j_i}+\tau^i_{j_i}-c_{j_i}||\\
&\leq& \sum_{p\in Opt_i}||p-m_i||+|Opt_i| \times (||m_i-\tau^i_{j_i}||+||\tau^i_{j_i}-c_{j_i}||)
%$$\leq\sum_{p\in Opt_i}||p-m_i||^2+|Opt_i|*2(||m_i-\tau^i_{j_i}||^2+||\tau^i_{j_i}-c_{j_i}||^2)$$
\end{eqnarray*}

By Lemma \ref{lem-closemedian}, we have
\begin{eqnarray*}
||\tau^i_{j_i}-m_i|| &\leq& \frac{2+\epsilon}{\frac{1}{k}}(\frac{1}{|Opt_i|}\sum_{p\in Opt_i}||p-m_i ||);\\
||\tau^i_{j_i}-c_{j_i}|| &\leq& \frac{2+\epsilon}{\frac{|\Gamma^i_{j_i}|}{|S_{j_i}|}}(\frac{1}{|S_{j_i}|}\sum_{p\in S_{j_i}}||p-c_{j_i} ||).
\end{eqnarray*}

%Meanwhile, also from Lemma \ref{lem-closemedian}, we have:
%$$||\tau^i_{j_i}-c_{j_i}||\leq \frac{2+\epsilon}{\frac{|\Gamma^i_{j_i}|}{|S_{j_i}|}}(\frac{1}{|S_{j_i}|}\sum_{p\in S_{j_i}}||p-c_{j_i} ||).$$
%

From the above inequalities, we have
%With these two inequalities for $||\tau^i_{j_i}-m_i||$ and $||\tau^i_{j_i}-c_{j_i}||$, we have:
\begin{eqnarray*}
\sum_{p\in Opt_i}||p-c_{j_i}|| &\leq & \sum_{p\in Opt_i}||p-m_i||+|Opt_i| \times (||m_i-\tau^i_{j_i}||+||\tau^i_{j_i}-c_{j_i}||)\\
&\leq&\sum_{p\in Opt_i}||p-m_i||^2+|Opt_i|(\frac{2+\epsilon}{\frac{1}{k}}(\frac{1}{|Opt_i|}\sum_{p\in Opt_i}||p-m_i ||)+\frac{2+\epsilon}{\frac{|\Gamma^i_{j_i}|}{|S_{j_i}|}}(\frac{1}{|S_{j_i}|}\sum_{p\in S_{j_i}}||p-c_{j_i} ||))\\
&=&((2+\epsilon)k+1)\sum_{p\in Opt_i}||p-m_i||+(2+\epsilon)\frac{|Opt_i|}{|\Gamma^i_{j_i}|} \times \sum_{p\in S_{j_i}}||p-c_{j_i} ||).
\end{eqnarray*}
Since $|\Gamma^i_{j_i}|\geq \frac{1}{k}|Opt_i|$,  we have $\frac{|Opt_i|}{|\Gamma^i_{j_i}|}\leq k$.  Thus, 
\begin{eqnarray*}
\sum_{p\in Opt_i}||p-c_{j_i}||^2\leq ((2+\epsilon)k+1)\sum_{p\in Opt_i}||p-m_i||+(2+\epsilon)k\sum_{p\in S_{j_i}}||p-c_{j_i} ||.
\end{eqnarray*}
Summing both sides of the above inequality over $i$, we have
\begin{eqnarray}
\sum^k_{i=1}\sum_{p\in Opt_i}||p-c_{j_i}||^2 &\leq& ((2+\epsilon)k+1)\sum^k_{i=1}\sum_{p\in Opt_i}||p-m_i||+(2+\epsilon)k\sum^k_{i=1}\sum_{p\in S_{j_i}}||p-c_{j_i} || \nonumber\\
&\leq& ((2+\epsilon)k+1)\sum^k_{i=1}\sum_{p\in Opt_i}||p-m_i||+(2+\epsilon)k^2 \sum^k_{j=1}\sum_{p\in S_{j}}||p-c_{j} ||. \label{for-13} 
\end{eqnarray}
It is easy to know that the optimal objective value of $k$-medians is  no larger than that of   $k$-CMedians on the same set of points in $\mathcal{G}$. Thus, $\sum^k_{j=1}\sum_{p\in S_{j}}||p-c_{j} ||\leq c\sum^k_{i=1}\sum_{p\in Opt_i}||p-m_i||$. Plugging this inequality into (\ref{for-13}), we have
\begin{eqnarray*}
\sum^k_{i=1}\sum_{p\in Opt_i}||p-c_{j_i}||\leq ((2+\epsilon)ck^2+(2+\epsilon)k+1)\sum^k_{i=1}\sum_{p\in Opt_i}||p-m_i||.
\end{eqnarray*}
The above inequality means that if we take the $k$-tuple $(c_{j_1}, \cdots, c_{j_k})$ as the $k$ approximate median points for $k$-CMedians, we have a $((2+\epsilon)ck^2+(2+\epsilon)k+1)$-approximation solution for $k$-CMedians. Thus, the theorem is proved.
% compute the minimum weight bipartite matching for each $G_i$ to $(c_{j_1}, \cdots, c_{j_k})$, we would get $((2+\epsilon)ck^2+(2+\epsilon)k+1)$-approximation for $k$-CMedians on $\mathcal{G}$. So the theorem is %proved. 
\qed
\end{proof}

\subsection{Peeling Algorithm for $k$-CMedians}
\label{sec-ptasmedian}

The following lemma is a key to the peeling algorithm for $k$-CMedians (i.e., play a similar role as 
%serves a similar purpose for $k$-Medians as 
Lemma \ref{lem-simplex} for $k$-CMeans).
%, we also have the following lemma for $k$-CMedians, which would help us in the peeling algorithm.

\begin{lemma}
\label{lem-simplexmedian}
Let $P$ to be a set of points in $\mathbb{R}^d$ with a partition $P=\cup^j_{l-1} P_l$,  $o$ be its optimal median point, and $o_l$ be the optimal median point of $P_l$ for $1\leq l\leq j$. 
Let $\mu=\frac{1}{|P|}\sum_{p\in P}||p-o||$. Then, there exists some $i_{0}$ such that $||o-o_{i_{0}}||\leq 4\mu$.
%If $\{m'_{1}, \cdots, m'_{j+1}\}$ satisfies the conditions of $||m'_{l}-m_{l}||^2\leq\epsilon\delta^2$ for $1\leq l\leq j$ and $||m'_{j+1}-m_{j+1}||^2\leq\epsilon\delta^{2}_{j+1}$,
%there exists an $O(d(2j/\epsilon)^j)$-time algorithm to generate an approximate mean point $m'$ of $Q$ in $\mathbb{R}^d$ such that $||m'-m||^2\leq O(\epsilon)\delta^2$. 
%The running time is $O(d(2j/\epsilon)^j)$.
\end{lemma}
\begin{proof}
Since $\mu=\frac{\sum_{p\in P}||p-o||}{|P|}=\sum^{l}_{i=1}(\frac{|P_{i}|}{|P|}\frac{\sum_{p\in P_{i}}||p-o||}{|P_{i}|})$, there must exist some index $i_{0}$ such that $\frac{\sum_{p\in P_{i_{0}}}||p-o||}{|P_{i_{0}}|}$ $\leq\mu$. By Markov inequality, we know that there exists one subset $U$ of $P_{i_0}$ such that $|U|> |P_{i_{0}}|/2$ and  $||p-o||\leq 2\mu$ for any $ p\in U$.

Since $o_{i_{0}}$ is the optimal median point of $P_{i_{0}}$,  $\frac{\sum_{p\in P_{i_{0}}}||p-o_{i_{0}}||}{|P_{i_{0}}|}\leq\frac{\sum_{p\in P_{i_{0}}}||p-o||}{|P_{i_{0}}|}\leq\mu$. Similarly, by Markov inequality, we know that there exists one subset $V$ of $P_{i_0}$ such that $|V|> |P_{i_{0}}|/2$ and  $||p-o_{i_0}||\leq 2\mu$ for any $ p\in V$.

From the inequalities of  $|U|> |P_{i_{0}}|/2$ and $|V|> |P_{i_{0}}|/2$ and the fact that $U\cap V \neq\emptyset$, we know that there exists one point $p_{0}\in U\cap V$ such that $||p_{0}-o||\leq 2\mu$ and $||p_{0}-o_{i_{0}}||\leq 2\mu$. Thus $||o_{i_{0}}-o||\leq ||o_{i_{0}}-p_{0}||+||p_{0}-o||\leq 4\mu$.
\qed
\end{proof}

Before presenting our peeling algorithm, we still need the following lemma proved by  Badoiu {\em et al.} in \cite{BHI} for finding an approximate solution for $1$-median.

\begin{theorem}[\cite{BHI}]
\label{the-1med}
Let $P$ be a normalized set of $n$ points in $\mathbb{R}^d$ space, $1>\epsilon>0$, and $R$ be a random sample of $O(1/\epsilon^3\log1/\epsilon)$ points from $P$. Then one can compute, in $O(d2^{O(1/\epsilon^4)}\log n)$ time, a point-set  $S(P,R)$ of cardinality $O(2^{O(1/\epsilon^4)}\log n)$ , such that with constant probability (over the choices of $R$), there is a point $q\in S(P,R)$ such that $cost(q,P)\leq (1+\epsilon)med_{opt}(P,1)$.
\end{theorem}

\noindent\textbf{Algorithm $k$-CMedians}
\newline \textbf{Input:} $\mathcal{G}=\{G_1, \cdots, G_n\}$, $k\geq 2$ and an small constant $\epsilon>0$.
\newline \textbf{Output:} a $(5+\epsilon)$-approximation solution for $k$-CMedians on $\mathcal{G}$.
\begin{enumerate}
\item Run the $(1+\epsilon)$-approximation $k$-medians algorithm from \cite{KSS} on $\mathcal{G}$, and let  $\Omega$ be the obtained objective value.

\item For $i=1$ to $\frac{4k}{\epsilon}$ do 
\begin{enumerate}
\item Set $\delta=\frac{\Omega}{4k}+i\frac{\epsilon}{4k}\Omega$, and run Algorithm Sphere-Peeling-Tree-2.

\item Let  $\mathcal{T}_i$ be the returned tree.

\end{enumerate}

\item For each path of every $\mathcal{T}_i$, use bipartite matching procedure to compute the objective value of $k$-CMeans on $\mathcal{G}$. Output the $k$ points from the path with smallest objective value.

\end{enumerate}

%The following {\em Peeling Tree Algorithm $2$} is called by Algorithm $2$.\\

\noindent\textbf{Algorithm Sphere-Peeling-Tree-$2$}
\newline \textbf{Input:} $\mathcal{G}$, $k\geq 2$, $\epsilon, \delta>0$.
\newline \textbf{Output:} A tree $\mathcal{T}$ of height $k$ with each node $v$ associated with a point $p_v \in \mathbb{R}^d$.
\begin{enumerate}
%\item Let  $\mathcal{R}=\cup^{\log(kn)}_{t=0}\{\frac{1+l\frac{\epsilon}{2}}{2(1+\epsilon)}j2^{t/2}\sqrt{\epsilon}\delta\mid 0\le l\le 4+\frac{2}{\epsilon}\}$ be the set of radius candidates.
\item Initialize $\mathcal{T}$ with a single root node $v$ associating with no point.
\item Recursively grow each node $v$ in the following way
\begin{enumerate}
\item If the height of $v$ is already $k$, then it is a leaf node.
\item Otherwise, let $j$ be the height of $v$. Build the set of radius candidates $\mathcal{R}=\cup^{\log(kn)}_{t=0}\{\frac{1+l\frac{\epsilon}{2}}{2(1+\epsilon)}j2^{t/2}\sqrt{\epsilon}\delta\mid 0\le l\le 4+\frac{2}{\epsilon}\}$.%enumerate all radius candidates in $\mathcal{R}$, and  
      For each radius candidate $r\in\mathcal{R}$ do
\begin{enumerate}
\item  Let $j$ be the height of $v$, and $\{p_{v_1}, \cdots, p_{v_j}\}$ be the $j$ points associated with nodes on the root-to-$v$ path  (including $p_v$). 

\item For each $p_{v_l}$, $1\leq l\leq j$, construct a ball $B_{j,l}$ centered at $p_{v_l}$ and with radius $r$. 
\item Take a random sample from $\mathcal{G}\setminus\cup^j_{l=1}B_{j,l}$ with size $m=\frac{8k^3}{\epsilon^9}\ln\frac{k^2}{\epsilon^6}$. Compute the approximate median points of all subsets of the sample (by Theorem \ref{the-1med}), and denote the set of the approximate median points as $\Pi$. Clearly, $|\Pi|=2^{m+O(1/\epsilon^4)}\log n$.
%\item For each $\pi_i \in \Pi$, construct the simplex determined by $\{p_{v_1}, \cdots, p_{v_j}, \pi_i\}$. Also construct the simplex determined by $\{p_{v_1}, \cdots, p_{v_j}\}$. For each simplex, construct the grid inside the simplex with size $O((\frac{4j}{\epsilon^2})^j)$. %as Theorem \ref{the-constant}. 

\item For each point $p$ in $\Pi$, add one child to $v$, and associate it with $p$; add another $j$ children, with each one associating with a different point in $\{p_{v_1}, \cdots, p_{v_j}\}$.
\end{enumerate}
\end{enumerate}

\end{enumerate}

We can use a similar approach as in Section \ref{sec-proof} to analyze the correctness of Algorithm $k$-CMedians. \\
%The only difference is that the induction assumption with Lemma \ref{lem-induction}. 

Let $\mathcal{OPT}=\{Opt_1, \cdots, Opt_k\}$ be the optimal solution of $k$-CMedians on $\mathcal{G}$.
%, we denote $\mathcal{G}=\{G_1, \cdots, G_n\}$ to be the instance of $k$-CMedians we want to solve, and the optimal solution as $\mathcal{OPT}=\{Opt_1, \cdots, Opt_k\}$. 
Without loss of generality, we assume that $|Opt_1|\ge |Opt_2|\ge\cdots \ge |Opt_k|$. For each $Opt_j$, $1\leq j\leq k$, let $m_j$ be its median point,  $\beta_j$ be its fraction in $\mathcal{G}$ (i.e., $|Opt_{j}|/|\cup_{i=1}^{n} G_{i}|$), and $\mu_j=\frac{1}{|Opt_j |}\sum_{p\in Opt_j}||p-m_j||$. Thus, $\beta_1\geq\cdots\geq\beta_k$ and $\sum^k_{j=1}\beta_j =1$. Also, let $\mu_{opt}=\sum^k_{j=1}\beta_j\mu_j$.  

\begin{lemma}
\label{lem-inductionmedian}
Among all the trees generated in Algorithm $k$-CMedians, there exists one tree $\mathcal{T}_i$, which has a root-to-leaf path with each node $v_{j}$ at level $j$, $1\leq j\leq k$, on the path associating a point 
%such that for each , the point 
$p_{v_j}$ and satisfying the inequality 
%associate with the node $v_j$ at the $j$-th level satisfies:
$$||p_{v_j}-m_j ||\leq 4\mu_j+(1+\epsilon)j\frac{\epsilon}{\beta_j}\mu_{opt} .$$
\end{lemma}

%Correspondingly, we have the following lemma:
\begin{lemma}
\label{lem-equalmedian}
If Lemma \ref{lem-inductionmedian} is true, Algorithm $k$-CMedians yields a $(5+O(k^2)\epsilon)$-approximation solution for $k$-CMedians.
\end{lemma}
\begin{proof}
We first assume that Lemma \ref{lem-inductionmedian} is true. Then for each $1\leq j\leq k$, we have
\begin{eqnarray}
\sum_{p\in Opt_j}||p-p_{v_j}|| & \le & \sum_{p\in Opt_j}||p-m_j||+|Opt_j|\times||m_j-p_{v_j}|| \nonumber\\
&\leq& \sum_{p\in Opt_j}||p-m_j||+|Opt_j|\times(4\mu_j+(1+\epsilon)j\frac{\epsilon}{\beta_j}\mu_{opt}) \nonumber\\
&=& 5|Opt_j |\mu_j+(1+\epsilon) j\epsilon|\mathcal{G}|\mu_{opt} \label{for-14}
\end{eqnarray}

%In above, the first equality follows from Lemma \ref{lem-meanshift} (note $m_j$ is the mean point of $Opt_j$). And the last inequality is from the fact that $(a+b)^2\leq 2(a^2+b^2)$ for any two real numbers $a$ and $b$.

Summing the both sides of  (\ref{for-14}) over $j$, we have
\begin{eqnarray}
\sum^k_{j=1}\sum_{p\in Opt_j}||p-p_{v_j}||^2 &\leq & \sum^k_{j=1}(5|Opt_j |\mu_j+(1+\epsilon) j\epsilon|\mathcal{G}|\mu_{opt}) \nonumber\\
&\leq& 5\sum^k_{j=1}|Opt_j |\mu_j+(1+\epsilon) k^2\epsilon|\mathcal{G}|\mu_{opt} \nonumber\\
&=&(5+O(k^2)\epsilon)|\mathcal{G}|\mu_{opt}. \label{for-15}
\end{eqnarray}

In the above, the last equation follows from the fact that $\sum^k_{j=1}|Opt_j |\mu_j=|\mathcal{G}|\mu_{opt}$. By (\ref{for-15}), we know that $\{p_{v_1}, \cdots, p_{v_k}\}$ induces a $(5+O(k^2)\epsilon)$-approximation solution for $k$-CMedians.
%via using bipartite matching algorithm. 
\qed
\end{proof}

By a similar argument given in the proof of  Lemma \ref{lem-induction}, we can show the correctness of Lemma \ref{lem-inductionmedian}. Thus, we have the following theorem.
\begin{theorem}
\label{the-ptasmedian}
With constant probability, Algorithm $k$-CMedians yields a $(5+\epsilon)$-approximation for $k$-CMedians in $O(2^{poly(\frac{k}{\epsilon})}n(\log n)^{2k} d )$ time. 
%where $f$ is a polynomial function of $\frac{1}{\epsilon}, k$.
\end{theorem}


\newpage
\begin{thebibliography}{7}
%\vspace{-0.1in}

\bibitem{ADH}
Esther M. Arkin, JosŽ Miguel D'az-B‡–ez, Ferran Hurtado, Piyush Kumar, Joseph S. B. Mitchell, BelŽn Palop, Pablo PŽrez-Lantero, Maria Saumell, Rodrigo I. Silveira: Bichromatic 2-Center of Pairs of Points. LATIN 2012: 25-36


\bibitem{AP98}
P. K. Agarwal and C. M. Procopiuc. ``Exact and Approximation Algorithms for Clustering,''{\em Proc. 9th ACM-SIAM Sympos. Discrete Algorithms, pages 658-667,} 1998.


%\bibitem{AS92}
%N. Alon and J. H. Spencer. '' The Probabilistic Method''. {\em John Wiley and Sons }, 1992.

%\bibitem{AKK}
%Sanjeev Arora, David R. Karger, Marek Karpinski, '' Polynomial time approximation schemes for dense instances of NP-hard problems''. {\em STOC 1995: 284-293}
%

\bibitem{AV07}
David Arthur, Sergei Vassilvitskii: ''k-means++: the advantages of careful seeding''. {\em SODA 2007: 1027-1035}


\bibitem{BBC}
Nikhil Bansal, Avrim Blum, Shuchi Chawla: ''Correlation Clustering''.{\em Machine Learning 56(1-3)}: 89-113 (2004)


%\bibitem{BC03}
%M.Badoiu, K.Clarkson, ``Smaller core-sets for balls", {\em Proceedings
%of the fourteenth annual ACM-SIAM symposium on Discrete algorithms}, pp.~801--802, 2003
%
\bibitem{BD06}
S.Basu, Ian Davidson: Clustering with Constraints Ð Theory and Practice. {\em ACM KDD 2006}


\bibitem{BHI}
M.Badoiu, S.Har-Peled, P.Indyk, ``Approximate clustering via core-sets", {\em Proceedings of the 34th Symposium on Theory of Computing}, pp.~250--257, 2002.

\bibitem{BKK}
C.Böhm, K.Kailing, P.Kröger, A.Zimek, "Computing Clusters of Correlation Connected Objects".{\em  Proc. ACM SIGMOD International Conference on Management of Data (SIGMOD'04), Paris, France. pp. 455–467. doi:10.1145/1007568.1007620. }


%\bibitem{CGW}
%Moses Charikar, Venkatesan Guruswami, and Anthony Wirth. ''Clustering with qualitative information''. {\em J. Comput. Syst. Sci., 71(3):360-383}, 2005.
%

%
\bibitem{D08}
S. Dasgupta, ''The hardness of k-means clustering''. {\em Technical Report}, 2008.

\bibitem{DEF}
Erik Demaine, Dotan Emanuel, Amos Fiat, and Nicole Immorlica. ''Correlation clustering in general weighted graphs''. {\em Theor. Comput. Sci., 361(2):172-187}, 2006

\bibitem {HDX11}
H.Ding, J.Xu, ''Solving Chromatic Cone Clustering via Minimum Spanning Sphere'', {\em ICALP}, 2011

%\bibitem{DX11}
%H.Ding, J.Xu, ''Chromatic Kernel and Its Applications'', {\em manuscript}, 2011

%\bibitem{FG88}
%T. Feder and D. H. Greene. ``Optimal algorithm for approximation
%clustering,'' {\em Proc. 20th ACM Symp. Theory of Computing, pages
%434-444,} 1988.



%\bibitem{G85}
%T. Gonzalez. ``Clustering to minimize the maximum intercluster
%distance,'' {\em Theoritical Computer Science, Vol. 38, pages 293-306,}
%1985.


\bibitem{GG06}
Ioannis Giotis and Venkatesan Guruswami. '' Correlation clustering with a fixed number of clusters''. {\em Theory of Computing, 2(1):249-266}, 2006.

%\bibitem{GW95}
%M. X. Goemans, D. P. Williamson, '' Improved Approximation Algorithms for Maximum Cut and Satisfiability Problems Using Semidefinite Programming''. {\em J. ACM 42(6): 1115-1145 (1995)}
%


\bibitem{HM04}
S. Har-Peled and S. Mazumdar, ``Coresets for k-Means and k-Median
Clustering and their Applications,'' {\em Proc. 36th ACM Symposium on
Theory of Computing, pages 291-300,} 2004.


%\bibitem{HSS}
%E.Hazan, S.Safra, O.Schwartz, '' On the Complexity of Approximating k-Dimensional Matching''.{\em RANDOM-APPROX 2003: 83-97}
%

\bibitem{IKI}
Mary Inaba, Naoki Katoh, Hiroshi Imai, '' Applications of Weighted Voronoi Diagrams and Randomization to Variance-Based k-Clustering (Extended Abstract)''. {\em Symposium on Computational Geometry 1994: 332-339}


%\bibitem{JZM}
%A. K. Jain, Y. Zhou, T. Mustufa, E. C. Burdette, G. S. Chirikjian, and G. Fichtinger, ``Matching and reconstruction of brachytherapy seeds using the Hungarian algorithm,'' {\em Proc. SPIE Med. Imag.: Visualizat., Image-Guided Procedures, Display, vol. 5744, pp. 810-821}, 2005

\bibitem{KR99}
S. G. Kolliopoulos and S. Rao, ``A nearly linear-time approximation
scheme for the euclidean k-median problem,'' {\em Proc. 7th Annu.
European Sympos. Algorithms, pages 378-389,} 1999.

\bibitem{KSS}
A. Kumar, Y. Sabharwal, S. Sen, '' Linear-time approximation
 schemes for clustering problems in any dimensions''. {\em J. ACM 57(2):2010}

%\bibitem {M90}
%N.Megiddo, "On the Complexity of Some Geometric Problems in Unbounded Dimension", {\em J. SYMB. COMPUT}, 10, pp. 327-334, 1990.
%
\bibitem{OSS}
R.Ostrovsky, Y.Rabani, L.J.Schulman, and C.Swamy. "The Effectiveness of Lloyd-Type Methods for the k-Means Problem". {\em Proceedings of the 47th Annual IEEE Symposium on Foundations of Computer Science (FOCS'06). pp. 165–174.}

%\bibitem{S04}
%Chaitanya Swamy. '' Correlation clustering: maximizing agreements via semidefinite programming''. {\em SODA'04: Procs. of the 15th ACM-SIAM Symposium on Discrete Algorithms, pages 526-527}, 2004.


%\bibitem{SMX}
%V.Singh, L.Mukherjee, J.Xu, K.R. Hoffmann, P. M. Dinu, M. Podgorsak, ''Brachytherapy Seed Localization Using Geometric and Linear Programming Techniques''. {\em IEEE Trans. Med. Imaging 26(9): 1291-1304 (2007) }
%
%\bibitem{Sin10}
%V. Singh, L. Mukherjee, J. Peng, and Jinhui Xu, ``Ensemble Clustering Using Semidefinite Programming with Applications." {\em Machine Learning}, Vol.79 No.1-2, pp.~177-200, 2010.
%
%%\bibitem{WB99}
%%West, D. Brent, ''Introduction to Graph Theory'', {\em Prentice Hall, Chapter 3, ISBN 0-13-014400-2}, 1999
%%
%%\bibitem{WSD}
%%X. Wei, J. Samarabandu, R.S. Devdhar, A.J. Siegel, R. Acharya and R.
%%Berezney,''Segregation of transcription and replication sites into
%%higher order domains,''{\em Sciences, 281:1502-6}, 1998.
%
%\bibitem{XX10}
%G. Xu, J. Xu, '' Efficient approximation algorithms for clustering
%point-sets.'' {\em Comput. Geom. 43(1): 59-66 }, 2010


\end{thebibliography}
 \end{document}